\PassOptionsToPackage{nosumlimits,nointlimits,nonamelimits}{amsmath}
\PassOptionsToPackage{final}{graphicx}
\PassOptionsToPackage{colorlinks,linkcolor={blue},citecolor={blue},urlcolor={red},breaklinks=true,final}{hyperref}

\documentclass[a4paper,UKenglish, autoref, thm-restate, final]{lipics-v2021}

\hideLIPIcs

\usepackage{etoolbox}

\renewcommand{\sectionautorefname}{Section}%
\usepackage{accents}
\newcommand\thickbar[1]{\accentset{\rule{.55em}{.6pt}}{#1}}

\usepackage{tikz-cd}

\tikzset{
  commutative diagrams/.cd,
  arrow style = tikz,
  diagrams    = {>=stealth},
  row sep     = large, 
  column sep  = huge
}

\usepackage{textcomp}

\BeforeBeginEnvironment{align}{\noindent\ignorespaces}

\usepackage{todos}
\usepackage{proof}

\input{catprog}

\usepackage{microtype}

\usepackage{amssymb}

\usepackage{mathrsfs}
\usepackage{microtype}

\newcommand{\klstar}{\sharp}  			%
\newcommand{\istar}{\dagger}  			%
\newcommand{\iistar}{\ddagger}  			%
\newcommand{\rstar}{\ast}  			%

\newcommand{\supp}{\oname{supp}}

\newcommand{\plus}{+}
\newcommand{\zero}{0}
\newcommand{\cpp}{\oplus}
\newcommand{\dc}{\mspace{\medmuskip};\mspace{.5\medmuskip}}

\renewcommand{\inl}{\oname{in_0}}
\renewcommand{\inr}{\oname{in_1}}

\newcommand{\inll}{\oname{in_{00}}}

\newcommand{\inlr}{\oname{in_{01}}}
\newcommand{\inrl}{\oname{in_{10}}}

\providecommand{\mto}{\mapsto}

\newcommand{\Mtx}{\catname{Mtx}}
\newcommand{\Fre}{\mathfrak{F}} %
\newcommand{\Frep}{\Fre} %
\newcommand{\Frepq}{\Fre^{?}} %

\renewcommand{\id}{1} %

\newcommand{\out}{\mathsf{out}}

\newcommand{\nat}{\mathbb{N}}

\newcommand{\Tame}[1][\BC]{\thickbar{#1}}
\newcommand{\Tst}{\BC^{?}}

\providecommand{\wt}[2]{\mfix{while}{\mathbin{}#1}{do}{\mathbin{}#2}}

\newcommand{\dec}{\mathinner{\raisebox{-1pt}{\scalebox{1.35}{$\diamond$}}}}

\newcommand{\Der}{\mathfrak{D}}

\makeatletter
\newcommand{\Bis}{\@ifstar{\@@Bis}{\@Bis}}
\newcommand{\@Bis}{\CB}
\newcommand{\@@Bis}{\mathrel{\CB}}
\makeatother

\providecommand{\autorefs}[1]{{\def\sectionautorefname{Sections}\autoref{#1}}}

\usepackage{savesym}

\savesymbol{emptyset}
\savesymbol{degree}
\savesymbol{leftmoon}
\savesymbol{rightmoon}
\savesymbol{fullmoon}
\savesymbol{newmoon}
\savesymbol{diameter}
\savesymbol{bigtimes}
\savesymbol{langle}
\savesymbol{rangle}
\savesymbol{triangleright}
\savesymbol{triangleleft}
\savesymbol{not}

\usepackage[matha]{mathabx}

\restoresymbol{other}{langle}
\restoresymbol{other}{rangle}
\restoresymbol{other}{triangleright}
\restoresymbol{other}{triangleleft}
\restoresymbol{other}{not}
\restoresymbol{other}{emptyset}

\providecommand{\mplus}{{\scriptscriptstyle\bf+}} 	       %

\let\div\undefined
\newcommand{\div}{\delta}

\makeatletter
\newcommand{\superimpose}[2]{{%
  \ooalign{%
    \hfil$\m@th#1\@firstoftwo#2$\hfil\cr
    \hfil$\m@th#1\@secondoftwo#2$\hfil\cr
  }%
}}
\makeatother

\newlength{\dhatheight}

\renewcommand{\hat}[1]{%
\mathchoice{
  \settoheight{\dhatheight}{\ensuremath{#1}}%
  \addtolength{\dhatheight}{.65\dhatheight}%
  \mathpalette\superimpose{
    {\kern.15pt\scalebox{.8}[.65]{\ensuremath{\widehat{\phantom{\rule{8pt}{\dhatheight}}}}}}%
    {\ensuremath{#1}}
  }
}{
  \settoheight{\dhatheight}{\ensuremath{#1}}%
  \addtolength{\dhatheight}{.6\dhatheight}%
  \mathpalette\superimpose{
    {\kern.15pt\scalebox{.8}[.65]{\ensuremath{\widehat{\phantom{\rule{8pt}{\dhatheight}}}}}}%
    {\ensuremath{#1}}
  }
}{
	\widehat #1
}{
	\widehat #1
}
}

\usepackage{xspace}

\newcommand{\axname}[1]{\text{\upshape\bfseries\textsf{#1}}\xspace}

\newcommand{\FIX}{\axname{Fixpoint}}
\newcommand{\NAT}{\axname{Naturality}}
\newcommand{\DIN}{\axname{Dinaturality}}
\newcommand{\UNI}{\axname{Uniformity}}

\newcommand{\COD}{\axname{Codiagonal}}

\newcommand{\Dec}{\BC^{\raisebox{-1pt}{\scalebox{.95}{$\diamond$}}}}

\newenvironment{ffigure}{
\begin{figure}[t]\begin{center}
\setlength{\fboxsep}{-5pt}
}{
\end{center}\end{figure}
}

\usepackage{lineno}
\linenumbers

\newcommand*\linenomathpatch[1]{%
  \cspreto{#1}{\linenomath}%
  \cspreto{#1*}{\linenomath}%
  \csappto{end#1}{\endlinenomath}%
  \csappto{end#1*}{\endlinenomath}%
}

\linenomathpatch{equation}
\linenomathpatch{gather}
\linenomathpatch{multline}
\linenomathpatch{align}
\linenomathpatch{alignat}
\linenomathpatch{flalign}

\usepackage{ifdraft} 
\ifdraft{
  \usepackage{showlabels} 
  
  \usepackage[layout=footnote,draft]{fixme}

}{
 \usepackage[layout=footnote,final]{fixme}
}

\FXRegisterAuthor{sg}{asg}{SG}	%

\renewcommand{\c}{\colon}

\newcommand{\bs}{\mathsf{b}}
\newcommand{\fs}{\mathsf{f}}

\newcommand{\us}{\mathsf{u}}

\newcommand{\lrule}[3]{\textbf{#1}~~\frac{#2}{#3}}

\makeatletter
\let\dir@frac\frac
\newcommand{\inv@frac}[2]{\dir@frac{#2}{#1}}

\renewcommand{\frac}{\@ifstar{\inv@frac}{\dir@frac}}
\makeatother

\usepackage{hypcap}
\newcommand{\boolnot}{\mathrel{\texttt{\upshape{\raisebox{0.5ex}{\,\texttildelow}}}}\kern-1.5pt}
\newcommand{\boolor}{\mathrel{\texttt{\upshape|\kern-2.95pt|}}}   %
\newcommand{\booland}{\mathrel{\texttt{\upshape\&\kern-.5pt\&}}} %
\renewcommand{\tt}{\oname{t\kern-.45pt t}}
\newcommand{\ff}{\oname{f\kern-.78pt f}}

\newcommand{\ifd}[3]{\mfix{\underline{if}}{\mathbin{}#1}{\underline{then}}{\mathbin{}#2}{\underline{else}}{\mathbin{\,}#3}}
\newcommand{\while}[2]{\mfix{\underline{while}}{\mathbin{\,}#1\mathbin{}}{\underline{do}}{\mathbin{\,}#2}}

\renewcommand{\paragraph}[1]{\medskip\noindent{\bfseries\rmfamily #1.}}

\title{A Unifying Categorical View of Nondeterministic Iteration and Tests}

\renewcommand{\c}{\colon}

\usepackage{hypcap}
\EventEditors{Rupak Majumdar and Alexandra Silva}
\EventNoEds{2}
\EventLongTitle{35th International Conference on Concurrency Theory (CONCUR 2024)}
\EventShortTitle{CONCUR 2024}
\EventAcronym{CONCUR}
\EventYear{2024}
\EventDate{September 9--13, 2024}
\EventLocation{Calgary, Canada}
\EventLogo{}
\SeriesVolume{311}
\ArticleNo{31}

\author{Sergey Goncharov}{Friedrich-Alexander-Universität Erlangen-Nürnberg, Germany}{sergey.goncharov@fau.de}{https://orcid.org/0000-0001-6924-8766}{German
  Research Foundation (DFG) project 501369690, Icelandic Research Fund project 228684-052}
\author{Tarmo Uustalu}{Reykjavik University, Iceland, and Tallinn University of Technology, Estonia}{tarmo@ru.is}{https://orcid.org/0000-0002-1297-0579}{Icelandic Research Fund project 228684-052}

\authorrunning{S.~Goncharov and T.~Uustalu}
\Copyright{Sergey Goncharov, Tarmo Uustalu} 
\ccsdesc[500]{Theory of computation~Categorical semantics}
\ccsdesc[500]{Theory of computation~Axiomatic semantics}

\keywords{Kleene iteration, Elgot iteration, Kleene algebra, coalgebraic resumptions}

\relatedversion{}
\relatedversiondetails{Full Version}{https://arxiv.org/abs/2407.08688}

\nolinenumbers %

\begin{document}\allowdisplaybreaks
\maketitle

\begin{abstract}
We study Kleene iteration in the categorical context. A celebrated completeness result
by Kozen introduced Kleene algebra (with tests) as a ubiquitous tool for lightweight reasoning about 
program equivalence, and yet, numerous variants of it
came along afterwards to answer the demand for more refined flavors of semantics, such 
as stateful, concurrent, exceptional, hybrid, branching time, etc. We detach \emph{Kleene iteration}
from \emph{Kleene algebra} and analyze it from the categorical perspective. The notion, we 
arrive at is that of \emph{Kleene-iteration category} (with coproducts and tests),
which we show to be general and robust in the sense of compatibility with programming 
language features, such as exceptions, store, concurrent behaviour, etc. We attest 
the proposed notion w.r.t.\ various yardsticks, most importantly, by characterizing 
the free model as a certain category of (nondeterministic) \emph{rational trees}.

\end{abstract}

\section{Introduction}

Axiomatizing notions of iteration both algebraically and categorically is a 
well-established topic in computer science where two schools of thought can be 
distinguished rather crisply: the first one is based on the inherently nondeterministic 
\emph{Kleene iteration}, stemming from the seminal work of Stephen Kleene~\cite{Kleene56} and deeply rooted 
in automata and formal language theory; the second one stems from another seminal 
work -- by Calvin Elgot~\cite{Elgot75} -- and is based on another notion of iteration,
we now call \emph{Elgot iteration}. The most well-known instance of Kleene iteration is the 
one that is accommodated in the algebra of regular expressions where~$a^\rstar$ represents $n$-fold 
compositions $a\cdots a$ and $n$ nondeterministically ranges over all naturals.
More abstractly, Kleene iteration is an operation of the following type:
\begin{align*}%
\lrule{}{p\c X\to X}{p^\rstar\c X\to X}
\end{align*} 
Intuitively, we think of $p$ as a program whose inputs and outputs range over $X$, 
and of~$p^\rstar$ as a result of composing $p$ nondeterministically many times with itself. 
Elgot iteration, in contrast, is agnostic to nondeterminism,
but crucially relies on the categorical notion of binary coproduct, and thus can only be implemented 
in categorical or type-theoretic setting. Concretely, the typing rule for Elgot iteration 
is
\begin{align}\label{eq:elg-iter}\tag{$\istar$}
\lrule{}{p\c X\to Y+X}{p^\istar\c X\to Y}
\end{align} 
That is, given a program that receives an input from $X$, and can output either to
$Y$ or to~$X$,~$p^\istar$ self-composes $p$ precisely as long as $p$ 
outputs to $X$.

A profound exploration of both versions 
of iteration and their axiomatizations in the categorical context, more precisely, in the 
context of \emph{Lavwere theories}, has been done by Bloom and \'Esik in a series 
of papers and subsumed in their monograph~\cite{BloomEsik93}. One outcome of this 
work is that in the context of Lawvere theories, in presence of nondeterminism,
Kleene iteration and Elgot iteration are essentially equivalent -- the ensuing 
theory was dubbed \emph{iteration grove theory}~\cite{BloomEsikEtAl93}. 
The existing analysis still does not cover certain aspects,
which we expressly address in our present work, most importantly, the following.
\begin{itemize}%
  \item Lavwere theories are only very special categories, while iteration is a common
  ingredient of semantic frameworks, which often involve it directly via an ambient category with coproducts,
  and not via the associated Lavwere theory. 
  \item Previous results on the equivalence of Elgot iteration and Kleene iteration do not address 
  the connection between control mechanisms involved in both paradigms:
  Elgot iteration fully relies on coproducts for making decisions whether to continue 
  or to end the loop, while Kleene iteration for the same purpose uses an additional mechanism of \emph{tests}~\cite{Kozen97a},
  which are specified axiomatically and thus yield a higher degree of flexibility. 
  \item A key feature of Kleene iteration of Kleene algebra,
  are the quasi-equational laws, which can be recast~\cite{Goncharov23} 
  to a form of the versatile and powerful \emph{uniformity principle}~(e.g.~\cite{SimpsonPlotkin00}).
  The latter is parameterized by a class of well-behaved elements, which in Kleene 
  algebra coincide with the algebra's entire carrier. However, in many situations, this class
  has to be restricted, which calls for axiomatizing it, analogously to tests.
\end{itemize}
Here, we seek a fundamental, general and robust categorical notion of 
Kleene iteration, which addresses these issues, is in accord with Elgot iteration
and the corresponding established laws for it (Elgot iteration operators that satisfy
these laws are called \emph{Conway operators}). In doing so, we depart from the 
laws of Kleene algebra, and relax them significantly. Answering the question how 
to do this precisely and in a principled way is the main insight of our work.

Let us dwell briefly on the closely related issues of generality and robustness. 
The laws of Kleene algebra, as originally axiomatized by Kozen~\cite{Kozen94}, capture a very 
concrete style of semantics, mirrored in the corresponding free model, which is 
the algebra of regular events, i.e.\ the algebra of regular sets of strings over 
a finite alphabet of symbols, with iteration rendered as a least fixpoint. 
Equations validated by this model are thus shared by the whole class of Kleene 
algebras. By regarding Kleene algebra terms as programs, the interpretation over 
the free model can be viewed as \emph{finite trace semantics} of linear-time nondeterminism. 
A standard example of properly more fine-grained -- branching-time -- nondeterminism is (bisimulation-based) process algebra, 
which fails the Kleene algebra's law of distributivity from the left:
\begin{align}\label{eq:ldist}
p\dc(q\plus r) =&\; p\dc q\plus  p\dc r.
\end{align}
Similarly, if we wanted to allow our programs to raise exceptions, the laws of Kleene 
algebra would undesirably force all exceptions to be equal:
\begin{align*}
\oname{raise} e_1 = \oname{raise} e_1\dc\zero = \zero = \oname{raise} e_2\dc\zero = \oname{raise} e_2.
\end{align*}
Here, we combine the law $p\dc\zero = \zero$ of Kleene algebra with the equation
$\oname{raise}_i\dc p = \oname{raise}_i$ that alludes to the common
programming knowledge that raising an exception exits the program instantly
and discards any subsequent fragment $p$. The resulting equality $\oname{raise} e_1 = \oname{raise} e_2$
states that raising exception $e_1$ is indistinguishable from raising exception~$e_2$.

We can interpret these and similar examples as evidence that the axioms of Kleene 
algebra are not sufficiently robust under extensions by programming language features. 
More precisely, Kleene algebras can be scaled up to \emph{Kleene monads}~\cite{GoncharovSchroderEtAl09},
and thus reconciled with Moggi's approach to computational effects~\cite{Moggi91a}. An important
ingredient of this approach are \emph{monad transformers}, which allow for combining
effects in a principled way. For example, one uses the \emph{exception monad transformer} 
to canonically add exception raising to a given monad. The above indicates that 
Kleene monads are not robust under this transformer.

Finally, even if we accept all iteration-free implications of Kleene algebra, these will not 
jointly entail the following identity:
\begin{align}\label{eq:star-idmp}%
\id^\rstar = \id,
\end{align}
which is however entailed by the Kleene algebra axioms. One setting where~\eqref{eq:star-idmp}
is undesirable is domain theory, which insists on distinguishing \emph{deadlock}
from \emph{divergence}, in particular, \eqref{eq:star-idmp} is failed by interpreting 
programs over the \emph{Plotkin powerdomain}~\cite{Plotkin76}. Intuitively, 
\eqref{eq:star-idmp} states that, if a loop \emph{may} be exited, it \emph{will} eventually be 
exited, while failure of \eqref{eq:star-idmp} would mean that the left program \emph{may} diverge, while the 
right program \emph{must} converge, and this need not be the same. Let us call the corresponding 
variant of Kleene algebra, failing~\eqref{eq:star-idmp}, \emph{may-diverge Kleene algebra}. However, it is not 
a priori clear how the axioms of may-diverge Kleene algebras must look like, given that 
\eqref{eq:star-idmp} is not a Kleene algebra axiom, but a consequence of the 
assumption that Kleene iteration is a least fixpoint. Hence, in may-diverge 
Kleene algebras Kleene iteration is not a least fixpoint (w.r.t.\ the order, induced by $\plus$). 

The notion we develop and present here is that of \emph{Kleene-iteration category (with tests) (KiC(T))}.
It is designed to address the above issues and to provide a uniform general 
and robust framework for Kleene iteration in a category. We argue in various ways 
that KiC(T) is in a certain sense the most basic practical notion of Kleene iteration,
most importantly by characterizing its free model, as a certain category of (nondeterministic)
rational trees.

\paragraph{Related work}
The (finite or $\omega$-complete) partially additive categories (PACs) by
Arbib and Manes \cite{ArbibManes80} and the PACs with effects of Cho \cite{Cho15} are
similar in spirit to KiCs in that they combine structured homsets and
coproducts, but significantly more special; in particular they support relational, 
sets of traces and similar semantics, but not branching time semantics. 
A PAC is a category with coproducts
enriched in partial commutative monoids (PMC). The PMC structure of
homsets and the coproducts are connected by axioms that make the PMC
structure unique. In an $\omega$-complete PAC, these axioms also
ensure the presence of an Elgot iteration operator, which is computed as a least fixpoint. A PAC with effects
comes with a designated effect algebra object; this defines a wide
subcategory of total morphisms, with coproducts inherited from the
whole category. Effectuses \cite{Jacobs15} achieve the same as PACs
with effects, but
starting with a category of total morphisms and then adding partial
morphisms. Cockett~\cite{Cockett07} recently proposed a notion of iteration in a category, based on restriction 
categories, and analogous to Elgot iteration \eqref{eq:elg-iter}, but avoiding
binary coproducts in favor of a suitably axiomatized notion of disjointness for morphisms.

In the strand of Kleene algebra, various proposals were made with utilitarian motivations
to weaken or modify the Kleene algebra laws, and thus to cope with process algebra~\cite{FokkinkZantema94}, branching 
behaviour~\cite{Moller07}, probability~\cite{McIverRabehajaEtAl11}, statefulness~\cite{GrathwohlKozenEtAl14}, 
graded semantics~\cite{GomesMadeiraEtAl17}, without however aiming to identify the conceptual 
core of Kleene iteration, which is our objective here. A recent move within this 
tendency is to eliminate nondeterminism altogether, with \emph{guarded
Kleene algebras} \cite{SmolkaFosterEtAl20}, which replace nondeterministic choice and
iteration with conditionals and while-loops. This is somewhat related to our analysis of 
tests and iteration via while-loops, but largely orthogonal to our main objective
to stick to Kleene iteration as nondeterministic operator in the original sense. 
Our aim to reconcile Kleene algebra, (co)products and Elgot iteration is rather 
close to that of Kozen and Mamouras~\cite{KozenMamouras13}.

Our characterization of the free KiCT in a way reframes the original Kozen's characterization
of the free Kleene algebra~\cite{Kozen94}. We are not generalizing this result though, essentially because we 
work in categories with coproducts, while a true generalization would only be achieved 
via categories without any extra structure (noting that algebras are single-object categories). This distinction becomes particularly important
in the context of branching~time semantics, which we also cover by allowing a controlled use of 
programs that fail distributivity from the left~\eqref{eq:ldist}. An axiomatization 
for such semantics has been proposed by Milner~\cite{Milner84} and was shown to be
complete only recently~\cite{Grabmayer22}. Again, we are not generalizing this 
result, since the definability issues, known to be the main obstruction for 
completeness arguments there, are not effective in presence of coproducts. 

\paragraph{Plan of the paper} We review minimal notations and 
conventions from category theory in \autoref{sec:prelim}. We then introduce 
idempotent grove and Kleene-Kozen categories in \autoref{sec:KK} to start off. 
In \autoref{sec:dec}, we formally compare two control mechanisms in categories: 
decisions and tests. In \autorefs{sec:kleene},~\ref{sec:while}, we establish 
equivalent presentations of nondeterministic iteration as Kleene iteration, as 
Elgot iteration and as while-iteration. In \autoref{sec:free} we construct a free
model for our notion of iteration, and then come to conclusions in \autoref{sec:conc}.

\section{Notations and Conventions}\label{sec:prelim}
We assume familiarity with the basics of category 
theory~\cite{Mac-Lane71,Awodey10}. In a category~$\BC$,~$|\BC|$ will denote the 
class of objects and $\BC(X,Y)$ will denote the set of morphisms from~$X$ to~$Y$. The judgement 
$f\c X\to Y$ will be regarded as an equivalent to $f\in\BC(X,Y)$ if $\BC$ is clear
from the context.
We tend to
omit indexes at natural transformations for readability. A subcategory~$\BD$ of 
$\BC$ is called \emph{wide} if $|\BC|=|\BD|$. We will use diagrammatic composition 
$\dc$ of morphisms throughout, i.e.\ given $f\c X\to Y$ and $g\c Y\to Z$, $f\dc g\c X\to Z$. 
We will denote by $\id_X$, or simply $\id$ the identity morphism on~$X$. 

\paragraph{Coproducts} In this paper, by calling $\BC$ ``a category with coproducts'' we will always mean that 
$\BC$ has \emph{selected binary coproducts}, i.e.\ that a bi-functor $\cpp\c\BC\times\BC\to\BC$ exists
such that $X\cpp Y$ is a coproduct of $X$ and $Y$. 
In such a category, we write 
$\inl\c X\to X\cpp Y$ and $\inr\c Y\to X\cpp Y$
for the left and right coproduct injections correspondingly. We will occasionally 
condense $\oname{in_{i}}\dc\oname{in_{j}}$ to $\oname{in_{i\,j}}$ for the sake of succinctness. 
\paragraph{Monads} A monad~$\BBT$ on~$\BC$ is determined by a \emph{Kleisli triple}~$(T,\eta, (-)^\klstar)$,
consisting of a map $T\c{|\BC|\to|\BC|}$, a family of morphisms 
$(\eta_X\c X \to TX)_{X\in|\BC|}$ and \emph{Kleisli lifting} sending each~$f\c X \to T Y$ to~$f^\klstar\c TX \to TY$ and
obeying \emph{monad laws}: %
\begin{align*}
\eta^{\klstar}=\id, && 
\eta\dc f^{\klstar}=f,  && 
(g\dc f^{\klstar})^{\klstar}=g^{\klstar}\dc f^{\klstar}.
\end{align*}
It follows that~$T$ extends to a functor, $\eta$ extends to a natural transformation -- \emph{unit},
${\mu = \id^\klstar}\c TTX\to TX$ extends 
to a natural transformation -- \emph{multiplication}, and that $(T,\eta,\mu)$ 
is a monad in the standard sense~\cite{Mac-Lane71}. We will generally use blackboard capitals (such as~$\BBT$) to refer to monads 
and the corresponding Roman letters (such as~$T$) to refer to their functor parts.
Morphisms of the form~$f\c X\to TY$ are called \emph{Kleisli morphisms} and form the \emph{Kleisli
category}~$\BC_{\BBT}$ of $\BBT$ under \emph{Kleisli composition} $f,g\mto f\dc g^\klstar$ 
with identity~$\eta$. If $\BC$ has binary coproducts then so does $\BC_\BBT$: the 
coproduct injections are Kleisli morphisms of the form $\inl\dc\eta\c X\to T(X\cpp Y)$,
$\inr\dc\eta\c Y\to T(X\cpp Y)$.

\paragraph{Coalgebras} Given an endofunctor $F\c\BC\to\BC$, a pair $(X\in |\BC|,c\c X\to FX)$
is called an $F$-coalgebra. Coalgebras form a category under the following notion
of morphism: $h\c X\to X'$ is a morphism from $(X,c)$ to $(X',c')$ if $h\dc c' = c\dc Fh$.
A terminal object in this category is called a \emph{final coalgebra}. We reserve 
the notation $(\nu F,\out)$ for a selected final coalgebra if it exists. A well-known 
fact (Lambek's lemma) is that $\out$ is an isomorphism.

For a coalgebra $(X, c\c X\to FX)$ on $\Set$ a relation $\Bis\subseteq X\times X$ is 
a \emph{(coalgebraic) bisimulation} if it extends to a coalgebra $(\Bis,b\c \Bis\to F\Bis)$, such 
that the left and the right projections from $\Bis$ to $X$ are coalgebra morphisms;
$x\in X$ and $y\in X$ are \emph{bisimilar} if $x\Bis*y$ for some bisimulation $\Bis$;
the coalgebra $(X, c\c X\to FX)$ is \emph{strongly extensional} \cite{TuriRutten98} 
if bisimilarity entails equality. Final coalgebras are the primary example of 
strongly extensional coalgebras.

\section{Idempotent Grove and Kleene-Kozen Categories}\label{sec:KK}
A monoid is precisely a single-object category. Various algebraic structures
extending monoids can be generalized to categories along this basic observation
(e.g.\ a group is a single-object groupoid, a quantale is a single-object quantaloid, etc.). 
In this section, we consider two classes of categories for nondeterminism and Kleene 
iteration, 
which demonstrate our principled 
categorical approach of working with algebraic structures.
\begin{definition}[Idempotent Grove Category, cf.~\cite{BensonTiuryn89,BloomEsikEtAl93}]
Let us call a category $\BC$ an \emph{idempotent grove category}
if the hom-sets of $\BC(X,Y)$ are equipped with the structure $(0,+)$ of bounded 
join-semilattice such that, 
for all $p\in\BC(Y,Z)$ and $q, r\in\BC(X,Y)$, 
\begin{align}\label{eq:dist_plus}
\qquad\qquad \zero\dc p =&\; \zero,& (q\plus r)\dc p =&\; q\dc p\plus r\dc p. &&
\intertext{
  In such a category, we call a morphism $p\in\BC(X,Y)$ \emph{linear} if it satisfies, for all $q,r\in\BC(Y,Z)$, 
} 
\label{eq:dist_plus'}
\qquad\qquad p\dc\zero =&\; \zero,& p\dc(q\plus r) =&\; p\dc q\plus  p\dc r. && %
\end{align}
An \emph{idempotent grove category with coproducts} is an idempotent grove 
category with selected binary coproducts and with $\inl$ and $\inr$ linear.
\end{definition}
Given $p,q\in\BC(X,Y)$, let $p\leq q$ if $p\plus q=q$. This 
yields a partial order with $\zero$ as the bottom element, and morphism composition  
is monotone on the left, while linear morphisms are additionally monotone on the right.
The class of all linear morphisms of an idempotent grove category thus forms 
a sub-category enriched in bounded join-semilattices (equivalently: commutative
and idempotent monoids) -- thus, an idempotent grove category where all morphisms are 
linear is an enriched category. However, we 
are interested in categories where not all morphisms are linear. An instructive example
is as follows.
\begin{example}[Synchronization Trees]\label{exa:ST}
Let $A$ be some non-empty fixed set of \emph{labels}, and let $TX=\nu\gamma.\,\CSet(X\cpp A\times\gamma)$
where $\CSet$ is the countable powerset functor. By generalities~\cite{Uustalu03},~$T$ 
extends to a monad $\BBT$ on~$\Set$. The elements of~$TX$ can be characterized 
as~\emph{countably-branching strongly extensional synchronization trees with exit 
labels in $X$}. Synchronization trees have originally been introduced by Milner~\cite{Milner80}
as denotations of process algebra terms, and subsequently generalized to 
infinite branching and to explicit exit labels (e.g.~\cite{AcetoCarayolEtAl12}). 
A generic element $t\in TX$ can be more explicitly represented using the following 
syntax:
\begin{align*}
	t=\sum_{i\in I} a_i.\,t_i+\sum_{i\in J} x_i
\end{align*}
where $I$ and $J$ are at most countable, the $a_i$ range over $A$, the 
$t_i$ range over~$TX$, and $x_i$ range over $X$. The involved summation operators 
$\sum_{i\in I}$ are considered modulo countable versions of associativity, 
commutativity and idempotence, and  $\zero = \sum_{i\in\emptyset} t_i$ and $t_1\plus t_2 = \sum_{i\in\{1,2\}} t_i$.

Recall that strong extensionality means that bisimilar elements are equal~\cite{TuriRutten98}. 
The Kleisli category of ${\BBT}$ is idempotent grove with $\zero$ and $\plus$ inherited from $\CSet$ and
ensuring~\eqref{eq:dist_plus} automatically.
It is easy to see that linear morphisms are precisely those that do not 
involve~actions.
\end{example}
A straightforward way to add a Kleene iteration operator to a category 
is as follows. 
\begin{definition}[Kleene-Kozen Category~\cite{Goncharov23}]\label{def:kk}
An idempotent grove category $\BC$ is a \emph{Kleene-Kozen category} if all morphisms of $\BC$
are linear and there is a \emph{Kleene iteration} operator
$(\argument)^\rstar\c\BC(X,X)\to\BC(X,X)$
such that, for any $p\c X\to X$, $q\c Y\to X$
and ${r\c X\to Z}$, the morphism $q\dc p^\rstar$ is the least (pre\dash)fixpoint of $q\plus (\argument)\dc p$
and the morphism $p^\rstar\dc r$ is the least (pre\dash)fixpoint of $r\plus p\dc(\argument)$.
\end{definition}
It is known~\cite{Goncharov23} that Kleene algebra is precisely a single-object 
Kleene-Kozen category.

In idempotent grove categories with coproducts, the following property is a direct consequence 
of linearity of $\inl$, $\inr$, and will be used extensively throughout. 
\begin{proposition}\label{pro:copr}
In idempotent grove categories with coproducts, $[p,q]\plus [p',q'] = {[p\plus p',q\plus q']}$.
\end{proposition}

\section{Decisions and Tests in Category}\label{sec:dec}
We proceed to compare two mechanisms for modeling control in categories:
\emph{decisions} and \emph{tests}. The first one is inherently categorical, and requires 
coproducts. The second one needs no coproducts, but requires nondeterminism. The 
latter one is directly inspired by tests of the Kleene algebra with tests~\cite{Kozen97a}. 
We will show that tests and decisions are in a suitable sense equivalent, when it comes 
to modeling control that satisfies Boolean algebra laws.

\begin{definition}[Decisions~\cite{CockettLack07,Goncharov23}]
In a category $\BC$ with binary coproducts, we call morphisms from $\BC(X,X\cpp X)$
\emph{decisions}.
\end{definition}
We consider the following operations on decisions, modeling truth values 
and logical connectives: $\tt=\inr$ (true), $\ff=\inl$ (false), $\boolnot d =d\dc[\inr,\inl]$ (negation),
$d\boolor e = d\dc[e,\inr]$ (disjunction), $d\booland e = d\dc[\inl,e]$ (conjunction).
Even without constraining decisions in any way, certain logical properties 
can be established, e.g.\ (not necessarily commutative or idempotent) monoidal 
structures $(\ff,\boolor)$, $(\tt,\booland)$, involutivity of $\boolnot$, ``de Morgan laws''
$\boolnot(d\boolor e) = \boolnot d\booland \boolnot e$, $\boolnot(d\booland e) = \boolnot d\boolor \boolnot e$, 
and the laws $\tt\boolor d = \tt$, $\ff\booland d = \ff$.

Given $d\in\BC(X,X\cpp X)$ and $p,q\in\BC(X,Y)$, let
\begin{align}\label{eq:ift}
\ifd{d}{p}{q} = d\dc [q,p].
\end{align}\par

\begin{definition}[Tests]\label{def:tests}
Given an idempotent grove category $\BC$, we call a family of linear morphisms
$\Tst={(\Tst(X)\subseteq\BC(X,X))_{X\in|\BC|}}$ \emph{tests} if every~$\Tst(X)$ 
forms a Boolean algebra under $\dc$ as conjunction and $\plus$ as disjunction.
\end{definition}
It follows that $\id\in\Tst(X)$ and $\zero\in\Tst(X)$ correspondingly are the top and bottom elements
of $\Tst(X)$. Given $b\in\Tst(X)$, $p,q\in\BC(X,Y)$, let
\begin{align}\label{eq:if}
\ift{b}{p}{q} = b\dc p\plus \bar b\dc q.
\end{align}
In an idempotent grove category $\BC$ with coproducts and tests $\Tst$, let~$?\c\BC(X,X\cpp X)\to\Tst(X)$ be the morphism
$d? = d\dc[\zero,\id].$
\begin{proposition}\label{pro:if_as_join}
Let $\BC$ be an idempotent grove category with coproducts. If a decision~$d$ is linear, then, for all $p$ and $q$, we have
  $\ifd{d}{p}{q} = \ift{d?}{p}{q}$. 
\end{proposition}
Let us say that a pair $(b,c)\in\BC(X,X)\times\BC(X,X)$ satisfies \emph{(the law of) 
contradiction} if $b\dc c=\zero$, and that it satisfies \emph{(the law of) excluded middle}
if $b\plus c=\id$.
The following characterization is instructive.
\begin{proposition}\label{pro:test_char}
Given an idempotent grove category $\BC$, a family of linear morphisms $\Tst={(\Tst(X)\subseteq\BC(X,X))_{X\in|\BC|}}$ 
forms tests for $\BC$ iff, for every $b\in\Tst$, there is $\bar b\in\Tst$ such that
$(b,\bar b)$ satisfies contradiction and excluded middle.
\end{proposition}
Note that the smallest choice of tests in $\BC$ is $\Tst(X) = \{\zero,\id\}$. We 
proceed to characterize the smallest possible choice of tests, sufficient for modeling 
control. 
\begin{definition}[Expressive Tests]\label{def:exp}
We call the tests $\Tst$ \emph{expressive} if every~$\Tst(X)$ contains~$[\inl,\zero]$ 
whenever $X=X_1\cpp X_2$.
\end{definition}
\begin{lemma}\label{lem:exp}
The smallest expressive family of tests always exists and is obtained by closing  
tests of the form $\zero$, $\id$, $[\inl,\zero]$, and $[\zero,\inr]$ under $\plus$ and $\dc$.
\end{lemma}
In the sequel, we will use the notation $\top$, $\bot$, $\land$, $\lor$ for tests, 
synonymously to $\id$, $\zero$, $\dc$, $\plus$ to emphasize their logical character.
\begin{lemma}\label{lem:tests-dec}
Let $\Tst$ be tests in an idempotent grove category $\BC$ with binary coproducts. 
\begin{enumerate}%
  \item The morphisms $\dec\c\Tst(X)\to\BC(X,X\cpp X)$, ${?\c\BC(X,X\cpp X)\to\Tst(X)}$ defined by
$\dec b = \bar b\dc\inl \plus b\dc\inr$, $d? = d\dc[\zero,\id]$
form a retraction.
\item Every morphism $d$ in the image of $\dec$ is linear. Moreover, we have
$d\dc\nabla = \id$, $d=d\booland d$, and $d=d\boolor d$.
\item For all $e$ and $d$ in the image of $\dec$, it holds that $(e\boolor d)? = e?\lor d?$, $(e\booland d)? = e?\land d?$, and
$(\boolnot d)?=\overline{d?}$.
\end{enumerate}
\end{lemma}
\autoref{lem:tests-dec} indicates that in presence of coproducts and with 
linear coproduct injections, instead of Boolean algebras on subsets of $\BC(X,X)$, 
one can equivalently work with Boolean algebras on subsets of $\BC(X,X\cpp X)$.

We conclude this section by an illustration that varying tests, in particular, 
going beyond smallest expressive tests is practically advantageous. 
\begin{example}
Consider the \emph{nondeterministic state monad} $\BBT$ with $TX = \PSet(S\times X)^S$
on~$\Set$, where $S$ is a fixed \emph{global store}, which the programs, represented
by Kleisli morphisms of $\BBT$ are allowed to read and modify. Morphisms of the Kleisli category $\Set_{\BBT}$ are equivalently (by uncurrying) 
maps of the form $p\c S\times X\to \PSet(S\times Y)$, meaning that $\Set_{\BBT}$
is equivalent to a full subcategory of $\Set_{\PSet}$, from which $\Set_{\BBT}$ inherits the 
structure of an idempotent grove category. The tests identified in \autoref{lem:exp}
are those maps $b\c S\times X\to \PSet(S\times X)$ that are determined by decompositions $X=X_1\cpp X_2$,
in particular, they can neither read nor modify the store. In practice, 
only the second is regarded as undesirable (and indeed would break commutativity of tests),
while reading is typically allowed. This leads to a more permissive notion
of tests, as those that are determined by the decompositions $S\times X=X_1\cpp X_2$.
\end{example}

\begin{ffigure}
\fbox{\parbox{\linewidth}{
\begin{align*}
&\inl\dc [p,q] = p\qquad
\inr\dc [p,q] = q\qquad
[\inl,\inr] = \id\qquad
[p,q]\dc r = [p\dc r,q\dc r]
\\[1ex]
&\zero\plus p = p\qquad
p\plus p = p\qquad
p\plus q = q\plus p\qquad
(p\plus q)\plus r = p\plus (q\plus r) 
\\[1ex]
&\zero\dc p = \zero\qquad 
(q\plus r)\dc p = q\dc p\plus r\dc p\qquad
u\dc\zero = \zero\qquad
u\dc(p\plus q) = u\dc p\plus  u\dc q 
\\[2ex]
  &(\axname{$\rstar$-Fix})~~ p^\rstar = \id \plus p\dc p^\rstar\quad(\axname{$\rstar$-Sum})~~ (p \plus  q)^\rstar =  p^\rstar\dc (q\dc p^\rstar)^\rstar\quad 
  (\axname{$\rstar$-Uni})~~\vcenter{
      \infer{u\dc p^\rstar = q^\rstar\dc u}{u\dc p = q\dc u}} %
  \end{align*}
}}
  \caption{Axioms of KiCs, including binary coproducts ($p,q,r$ range over $\BC$, $u$ ranges over $\Tame$).}
  \label{fig:kic}
\vspace{-2em}
\end{ffigure}

\section{Kleene Iteration, Categorically}\label{sec:kleene}
We now can introduce our central definition by extending idempotent grove categories with a selected class of linear morphisms, 
called \emph{tame morphisms}, and with Kleene iteration.
Crucially, we assume the ambient category $\BC$ to have coproducts as a necessary ingredient. 
Finding a general definition, not relying on coproducts, presently remains open.
\begin{definition}[KiC(T)]\label{def:GICCT}
We call a tuple $(\BC,\Tame)$ a \emph{Kleene-iteration category (KiC)}
if 
\begin{enumerate}%
  \item $\BC$ is an idempotent grove category with coproducts;
  \item $\Tame$ is a wide subcategory of $\BC$, whose morphisms we call \emph{tame}
    such that
    \begin{itemize}
      \item $\Tame$ has coproducts strictly preserved by the inclusion to $\BC$;
      \item  the morphisms of $\Tame$ are all linear;
    \end{itemize}    
  \item for every $X\in |\BC|$, there is a \emph{Kleene iteration} operator $(\argument)^\rstar\c\BC(X,X)\to\BC(X,X)$ 
  such that the laws~\axname{$\rstar$-Fix},~\axname{$\rstar$-Sum} and~\axname{$\rstar$-Uni} in \autoref{fig:kic}, with $u$ ranging over $\Tame$, are satisfied.
\end{enumerate}
A functor $F\c (\BC,\Tame)\to (\BD,\Tame[\BD])$ between KiCs 
is a coproduct preserving functor ${F\c\BC\to\BD}$ such that $F\zero = \zero$, $F(q\plus r) = Fq\plus Fr$ and $Fp^\rstar = (Fp)^\rstar$
for all $q,r\in\BC(X,Y)$, $p\in\BC(X,X)$, and $Fp\in\Tame[\BD](FX,FY)$ for all $p\in\Tame(X,Y)$. 

A KiC $(\BC,\Tame)$ equipped with a choice of tests $\Tst$ in $\Tame$ we call a \emph{KiCT} (=KiC with tests). 
Correspondingly, functors between KiCTs are additionally required to send tests to tests.  
\end{definition}
It transpires from the definition that the role of tameness is to limit the
power of the uniformity rule \axname{$\rstar$-Uni}. The principal case for $\BC\neq\Tame$ 
is~\autoref{exa:ST}. As we see later (\autoref{exa:resmp-iter}), this yields a KiC.
More generally, unless we restrict $\Tame$ to programs that satisfy the linearity laws~\eqref{eq:dist_plus'}, the uniformity principle \axname{$\rstar$-Uni} would tend to be unsound. Very roughly, uniformity is some infinitary form of distributivity from the left and it fails for programs that fail the standard left distributivity. This phenomenon is expected to occur for other flavors of concurrent semantics: as long as $\BC$ admits morphisms that fail~\eqref{eq:dist_plus'}, $\Tame$ would have to be properly smaller than $\BC$. Apart from concurrency, if $\Tame$ models a language with exceptions, those must be excluded from $\Tame$, for otherwise uniformity would again become unsound.

If we demand all morphisms 
to be tame, we will obtain a notion very close to that of Kleene-Kozen category (\autoref{def:kk}).
\begin{definition}[$\rstar$-Idempotence]
A KiC is \emph{$\rstar$-idempotent} if it satisfies~\eqref{eq:star-idmp}.
\end{definition}
\begin{proposition}\label{pro:kk}
A category $\BC$ is Kleene-Kozen 
iff $(\BC,\BC)$ is a $\rstar$-idempotent KiC.
\end{proposition}
\begin{proof}
As shown previously~\cite{Goncharov23}, $\BC$ is a Kleene-Kozen category iff 
\begin{enumerate}%
  \item $\BC$ is enriched over bounded join-semilattices and strict join-preserving morphisms;
  \item there is an operator $(\argument)^\rstar\c\BC(X,X)\to\BC(X,X)$
  such that 
  \begin{enumerate} %
    \item\label{it:klee-eq1} $p^\rstar = \id\plus p\dc p^\rstar$;
    \item\label{it:klee-eq2} $\id^\rstar = \id$; 
    \item\label{it:klee-eq4} $p^\rstar = (p\plus\id)^\rstar$;   
    \item\label{it:klee-eq3} $u\dc p = q\dc u$ ~implies~ $u\dc p^\rstar = q^\rstar\dc u$.
  \end{enumerate} 
\end{enumerate}
This yields sufficiency by noting that (1) states precisely that all morphisms 
in~$\BC$ are linear. To show necessity, it suffices to obtain~(2.c) from the assumptions 
that $(\BC,\BC)$ is a $\rstar$-idempotent KiC and that all morphisms in $\BC$ are 
linear. Indeed, we have $(p\plus\id)^\rstar = \id^\rstar\dc (p\dc\id^\rstar)^\rstar = p^\rstar$.\qed
\noqed\end{proof}
KiCs thus deviate from Kleene algebras precisely in four respects: 
\begin{enumerate}
  \item by generalizing from monoids to categories,
  \item by allowing non-linear morphisms, 
  \item by dropping $\rstar$-idempotence, and
  \item by requiring binary coproducts.
\end{enumerate}

\begin{example}
The axiom~\axname{$\rstar$-Sum}, included in~\autoref{def:GICCT}, is one of the 
classical \emph{Conway identities}. The other one 
$(p\dc q)^\rstar = \id\plus p\dc (q\dc p)^\rstar\dc q$
is derivable if $\BC=\Tame$, e.g.\ in Kleene-Kozen categories. 
Indeed, $q\dc p\dc q = q\dc p\dc q$ entails $q\dc (p\dc q)^\rstar = (q\dc p)^\rstar\dc q$ 
by~\axname{$\rstar$-Uni}, and using~\axname{$\rstar$-Fix}, $(p\dc q)^\rstar = 1\plus p\dc q\dc (p\dc q)^\rstar = 1\plus p\dc (q\dc p)^\rstar\dc q$.

Clearly, this argument remains valid with only $q$ being tame, but otherwise 
the requisite identity is not provable. 
\end{example}
It may not be obvious why the requirement to support binary coproducts is part of 
\autoref{def:GICCT}, given that the axioms of iteration do not involve them. The reason is that
certain identities that also do not involve coproducts are only derivable in their presence.
\begin{example}\label{exa:sqr}
The identity $p^\rstar = (p\dc (1\plus p))^\rstar$ holds in 
any KiC.
\end{example}
A standard way to instantiate \autoref{def:GICCT} is to start with a category 
$\BV$ with coproducts, and a monad $\BBT$ on it, and take $\BC=\BV_{\BBT}$,
$\Tame=\BV$ or, possibly, $\Tame=\BV_{\BBT}$. The monad must support nondeterminism and Kleene iteration 
so that the axioms of KiC are satisfied. Consider a class of Kleene-Kozen categories 
that arise in this way.

\begin{example}\label{exa:quantale}
Let $Q$ be a unital quantale, and let $TX=Q^X$ for every set $X$. Then $T$ extends 
to a monad on $\Set$ as follows: $\eta(x)(x) = 1$, $\eta(x)(y) = \bot$ if $x\neq y$, and
\begin{displaymath}
(p\c X\to Q^Y)^\klstar(f\c X\to Q)(y\in Y) = \bigor_{x\in X} p(x)(y)\cdot f(x).
\end{displaymath}
We obtain a Kleene-Kozen structure in $\Set_{\BBT}$ as follows: 
\begin{itemize}%
  \item $\zero\c X\to Q^Y$ sends $x$ to $\lambda y.\,\bot$;
  \item $p\plus q\c X\to Q^Y$ sends $x$ to $\lambda y.\, p(x)(y)\lor q(x)(y)$;
  \item $p^\rstar\c X\to Q^X$ is the least fixpoint of the map $q\mto \id\plus q\dc p$.
\end{itemize}
This construction restricts to $Q^X_{\omega_1} = \{f\c X\to Q\mid |\supp f|\leq\omega\}$
where ${\supp f}$ is the set of those $x\in X$, for which $f(x)\neq 0$. Thus,
e.g.\ the Kleisli categories of the powerset monad $\PSet$ and the countable 
powerset monad $\CSet$ are Kleene-Kozen.
\end{example}
For a contrast, consider a similar construction that yields a KiC, which 
is not Kleene-Kozen.
\begin{example}\label{exa:ini}
Let $Q=\{0,1,\infty\}$ be the complete lattice under the ordering $0<1<\infty$,
and let us define commutative binary multiplication as follows: $0\cdot x = 0$, 
$1\cdot x = x$ and $\infty\cdot\infty = \infty$. This turns $Q$ into a unital quantale,
hence an idempotent semiring, whose binary summation~$+$ is binary join. 
Next, define infinite summation with the formula
\begin{align*}
\sum_{i\in I} x_i =
\begin{cases}
\bigor_{i\in I'} x_i,&\text{if $I'=\{i\in I\mid x_i> 0\}$ is finite}\\
\infty,&\text{otherwise}
\end{cases}
\end{align*}
This makes $Q$ into a \emph{complete semiring}~\cite{DrosteKuich09}. Let us define 
the monad $\CR$ and the idempotent grove structure on $\Set_{\CR}$ like $Q^{(\argument)}_{\omega_1}$ in \autoref{exa:quantale} (with $\sum$ instead of $\bigor$).
For every ${f\c X\to \CR X}$, let $p^\rstar=\sum_{n\in\nat} p^n\c X\to \CR X$ where, 
inductively, $p^0= 1$ and $p^{n+1} = p \dc p^{n}$, and infinite sums are extended 
from $Q$ to the Kleisli hom-sets pointwise. 

It is easy to verify that $(\Set_{\CR},\Set_{\CR})$ is a KiC, but 
$\Set_{\CR}$ is not a Kleene-Kozen category, for $\rstar$-idempotence fails: 
	$\eta^\rstar = \sum_{n\in\nat} \eta^n = \sum_{n\in\nat} \eta =\lambda x,y.\,\infty\neq\eta$.
We will use a %
more convenient notation for the elements of $\CR X$ as infinite formal sums 
$\sum_{i\in I} x_i$ ($x_i\in X$), modulo associativity, commutativity, idempotence
(but without countable idempotence~$\sum_{i\in I} x = x$!).
\end{example}
Below, we provide two results for constructing new KiCs from 
old:~\autoref{the:grove-monad} and \autoref{the:resmp}, %
which are also used prominently in our characterization result in \autoref{sec:free}.
\begin{theorem}\label{the:grove-monad}
Let $(\BC,\Tame)$ be a KiC and let $\BBT$ be a 
monad on $\BC$ such that 
\begin{enumerate}
  \item $Tr^\rstar = (Tr)^\rstar$,
    for all $r\in\BC(X,X)$;
  \item the monad $\BBT$ restricts to a monad on $\Tame$.
\end{enumerate}
Then $\Tame_\BBT$ is a subcategory of $\Tame$ and $(\BC_\BBT,\Tame_\BBT)$ is a KiCT where $\zero$ and 
$\plus$ are defined as in $\BC$, and for any $p\c X\to TX$, the corresponding Kleene 
iteration is computed as $\eta\dc (p^\klstar)^\rstar$. 
\end{theorem}
Let us illustrate the use of~\autoref{the:grove-monad} by a simple example.
\begin{example}[Finite Traces]\label{exa:fin-tra}
Consider the monad $\PSet(A^\star\times\argument)$ on~$\Set$. Elements of 
$\PSet(A^\star\times X)$ are standardly used as (finite) trace semantics 
of programs. A trace is then a sequence of actions from $A$, followed by an end result 
in $X$. Of course, it can be verified directly that the Kleisli category of $\PSet(A^\star\times\argument)$
is Kleene-Kozen. Let us show how this follows from~\autoref{the:grove-monad}.  

The Kleisli category of $\PSet$ is isomorphic to the category of relations, and is 
obviously Kleene-Kozen. For $\PSet$, like for any commutative monad, the Kleisli category 
$\Set_{\PSet}$ is symmetric monoidal: $X\tensor Y=X\times Y$ and, given 
$p\c X\to\PSet Y$, $q\c X'\to\PSet Y'$, 
\begin{displaymath}
	(p\tensor q)(x,x') = \{(y,y')\mid y\in p(x)\comma y'\in q(x')\}.
\end{displaymath}
The set $A^\star$ is a monoid in $\Set_{\PSet}$ w.r.t.\ this monoidal structure. This 
yields a writer monad~$\BBT$ on $\Set_{\PSet}$ via $TX=A^\star\tensor X$ and $(Tp)(w,x) = \{(w,y)\mid y\in p(x)\}$. Its Kleisli category
$(\Set_{\PSet})_{\BBT}$ is isomorphic to our original Kleisli category of interest. The assumptions~(1)
of~\autoref{the:grove-monad} are thus satisfied in the obvious way. The assumption~(2)
is vacuous, as we chose all morphisms to be tame.
\end{example}
Finally, we establish robustness of KiCs under the 
\emph{generalized coalgebraic resumption monad transformer}~\cite{PirogGibbons13,GoncharovRauchEtAl15}, which is defined as follows.
\begin{definition}[Coalgebraic Resumptions]\label{def:resmp}
Let $\BV$ be a category with coproducts and let $\BBT$ be a monad on $\BV$.
Let $H\c\BV\to\BV$ be some endofunctor and assume that all final coalgebras $\nu\gamma.\,T(X\cpp H\gamma)$
exist. The assignment $X\mapsto \nu\gamma.\,T(X\cpp H\gamma)$ yields a monad $\BBT_H$, called
the (generalized) coalgebraic resumption monad transformer of $\BBT$.
\end{definition}

\begin{theorem}\label{the:resmp}
Let $\BBT_H$ be as in \autoref{def:resmp} and such that $(\BV_{\BBT},\Tame)$ 
is a KiC for some choice of $\Tame$. Then $\BV_{\BBT}$ is a 
wide subcategory of $\BV_{\BBT_H}$ and $(\BV_{\BBT_H},\Tame)$ 
is a KiC w.r.t.\ the following structure:
\begin{itemize}
  \item the bottom element in every $\BV(X,T_HY)$ is $\zero\dc\out^\mone$,
and the join of $p,q\in\BV(X,T_HY)$ is $(p\dc\out\plus q\dc\out)\dc \out^\mone$;
  \item given $p\in\BV(X,T_HX)$, $p^\rstar\in\BV(X,T_HX)$ is the unique 
  solution of the equation
  \begin{displaymath}
  	p^\rstar\dc\out = \inl\dc [p\dc\out,\zero]^\rstar\dc T(\id\cpp H p^\rstar).
  \end{displaymath}
\end{itemize}
\end{theorem}
We defer the proof to \autoref{sec:while} where we use reduction to the existing
result~\cite{GoncharovSchroderEtAl18}, using the equivalence of Kleene and Elgot iterations, 
we establish in \autoref{sec:while}. 
\begin{example}\label{exa:resmp-iter}
By taking $T=\CSet$ and $H=A\times\argument$ in \autoref{the:resmp}, we obtain $T_HX=\nu\gamma.\,\CSet(X\cpp A\times\gamma)$ 
from \autoref{exa:ST}. Let us illustrate the effect of Kleene iteration by example.
Consider the system of equations
\begin{align*}
P = a.\,P\plus Q,\qquad Q = b.\,Q\plus P
\end{align*}
for defining the behaviour of two processes $P$ and $Q$. This system induces a 
function $p\c\{P,Q\}\to T_H\{P,Q\}$, sending $P$ to $a.\,P\plus Q$ and $Q$ to $b.\,Q\plus P$. 
The expression $[p\dc\out,\zero]^\rstar$ calls the iteration operator of the powerset-monad,
resulting in the function that sends both~$P$ and $Q$ to $a.\,P\plus b.\,Q\plus Q\plus P$.
Finally, $p^\rstar$ sends $P$ to $P'$ and $Q$ to $Q'$, where $P'$ and~$Q'$
are the synchronization trees, obtained as unique solutions of the system:    
\begin{align*}
P' = a.\,P'\plus b.\,Q'\plus Q\plus P,\qquad Q = a.\,P'\plus b.\,Q'\plus Q\plus P.
\end{align*}
\end{example}

\section{Elgot Iteration and While-Loops}\label{sec:while}
In this section, we establish an equivalence between Kleene iteration, in the 
sense of KiC and \emph{Elgot iteration}, as an operation with the following 
profile in a category $\BC$ with coproducts:
\begin{align}\label{eq:elgot}
  (\argument)^\istar\c\BC(X,Y\cpp X)\to\BC(X,Y).
\end{align}
This could be done directly, but we prove an equivalence between Elgot 
iteration and while-loops first, and then prove the equivalence of the latter and 
Kleene iteration. In this chain of equivalences, only while-loops need tests,
and it will follow that a particular choice of tests is not 
relevant, once they are expressive. On the other hand, existence of expressive
tests is guaranteed by \autoref{lem:exp}. This explains why tests disappear 
in the resulting equivalence.
\begin{definition}[Conway Iteration, Uniformity]\label{defn:conway}
An Elgot iteration operator~\eqref{eq:elgot} in a category~$\BC$ with coproducts is 
\emph{Conway iteration}~\cite{Esik19} if it satisfies the following 
principles:
\begin{flalign*}
\axname{Naturality}:   &~~  p^{\istar}\dc q = (p\dc(q\cpp\id))^{\istar}          &&   %
~\axname{Dinaturality}: ~  (p\dc[\inl,q])^{\istar} = p\dc [\id,(q\dc[\inl,p])^{\istar}]  \\[1ex] %
\axname{Codiagonal}:   &~~  (p\dc[\id,\inr])^{\istar} = p^{\istar\istar}                       %
\end{flalign*}
Moreover, given a subcategory $\BD$ of $\BC$, $(\argument)^\istar$ is \emph{uniform}
w.r.t.\ $\BD$, or $\BD$-uniform, if it satisfies
\begin{flalign*}
\axname{Uniformity}:&\quad     \lrule{}{u\dc q = p\dc (\id\cpp u)}{u\dc q^{\istar} = p^{\istar}}\qquad\text{(with $u$ from $\BD$)} & 
\end{flalign*}
\end{definition}
By taking $q=\inr$ in \axname{Dinaturality}, we derive 
\begin{flalign*}
\axname{Fixpoint}: p\dc [\id, p^{\istar}] = p^{\istar}.
\end{flalign*}
\begin{ffigure}
\fbox{\parbox{\linewidth}{
\begin{flalign*}
\axname{DW-Fix:}  &\quad   \while{d}{p} = \ifd{d}{p\dc (\while{d}{p})}{\id}\\[1ex]
\axname{DW-Or:}   &\quad   \while{(d\boolor e)}{p} = (\while{d}{p})\dc \while{e}{(p\dc \while{d}{p})}\\[1ex]
\axname{DW-And:}  &\quad   \while{(d\booland (e\boolor\tt))}{p} = \while{d}{(\ifd{e}{p}{p})}\qquad\\[1.5ex]
\axname{DW-Uni:}  &\quad   \frac{u\dc \ifd{d}{p\dc \tt}{\ff} = \ifd{e}{q\dc u\dc \tt}{v\dc \ff}}
                            {u\dc \while{d}{p} = (\while{e}{q})\dc v}
\end{flalign*}
}}
  \caption{Uniform Conway iteration in terms of decisions.%
  }
  \label{fig:while}
\vspace{-2em}
\end{ffigure}
\begin{theorem}\label{thm:while}
Let $\BC$ be a category with coproducts, let $\BD$ be its wide subcategory with 
coproducts, preserved by the inclusion, and let for every 
${X\in|\BC|}$,~$\Dec(X)$ be a set of decisions such that 
(i) $\inl,\inr\in\Dec(X)$, (ii) $\Dec(X)$ is closed under~\eqref{eq:ift},
(iii) $\inl\cpp \inr\in\Dec(X)$ if $X=X_1\cpp X_2$.
Then, to give a $\BD$-uniform Conway iteration on $\BC$ is the 
same as to give an operator
\begin{align*}
\frac{d\in\Dec(X)\qquad p\in\BC(X,X)}{\while{d}{p}\in\BC(X,X)}
\end{align*}
that satisfies the laws in \autoref{fig:while} with $p,q$ ranging over $\BC$, and with $u,v$ ranging over~$\BD$.
\end{theorem}
\autoref{thm:while} yields an equivalence between two styles of iteration:
Elgot iteration and while-iteration. We next specialize it to idempotent grove categories
with tests using \autoref{lem:tests-dec}.

Note that in any KiC $(\BC,\Tame)$ with tests $\Tst$, in addition to
the if-then-else~\eqref{eq:if}, we have the while operator, defined in the standard way:
given $b\in\Tst(X)$, $p\in\BC(X,X)$,
\begin{align}\label{eq:if-while}
\wt{b}{p} = (b\dc p)^\rstar\dc\bar b.
\end{align}\par  
\begin{proposition}\label{pro:test_while}
Let $\BC$ be an idempotent grove category, let $\BD$ be a wide subcategory of~$\BC$ with 
coproducts, which are preserved by the inclusion to $\BC$, and with expressive tests
$\Tst$. Then $\BC$ supports $\BD$-uniform Conway 
iteration iff it supports a while-operator 
that satisfies the laws in~\autoref{fig:twhile}, where $b$ and $c$ come from 
$\Tst$, $p$ and~$q$ come from $\BC$ and $u,v$ come from~$\BD$ and the if-then-else operator is defined as in~\eqref{eq:if-while}.
\end{proposition}
\begin{proof}
For every $X\in |\BC|$, let $\Dec(X)$ be the image of $\Tst(X)$ under $\dec$ from \autoref{lem:tests-dec}.
As shown in the lemma, $\Dec(X)$ inherits the Boolean algebra structure from $\Tst(X)$. 
Using the isomorphism between $\Dec(X)$ and $\Tst(X)$ and \autoref{pro:if_as_join}, the laws 
from \autoref{fig:while} can be reformulated equivalently, resulting in~\axname{TW-Fix}, 
\axname{TW-Or}, \axname{TW-Uni}, and additionally
\begin{align*}
\wt{(b\land (c\lor\top))}{p} =&\; \wt{b}{(\ift{c}{p}{p})}
\end{align*}
which however holds trivially.
\end{proof}
Thus, in grove categories with expressive 
tests, Elgot iteration and while-loops are equivalent. We establish a 
similar equivalence between Kleene iteration and while-loops, which will  
entail an equivalence between Elgot iteration and Kleene iteration by transitivity.
\begin{theorem}\label{thm:while-kleene}
Let $\BC$, $\Tame$ and $\Tst$ be as follows.
\begin{enumerate}
  \item $\BC$ is an idempotent grove category with coproducts.
  \item $\Tame$ is a wide subcategory of $\BC$ with coproducts, 
  consisting of linear morphisms only and such that the inclusion of $\Tame$ to $\BC$ preserves coproducts.
  \item $\Tst$ are expressive tests in $\BC$.
\end{enumerate}
Then $(\BC,\Tame)$ is a KiCT iff $\BC$ supports a while-operator
 satisfying the laws in~\autoref{fig:twhile}. %
\end{theorem}
We can now characterize KiCs in terms of Elgot iteration.
\begin{theorem}\label{thm:elgot-grove}
Let $\BC$ be an idempotent grove category with coproducts, and let~$\Tame$ be
a wide subcategory of $\BC$ with coproducts, consisting of linear 
morphisms only and such that the inclusion of $\Tame$ to $\BC$ preserves coproducts.

Then $(\BC,\Tame)$ is a KiC iff $\BC$ supports $\Tame$-uniform 
Conway iteration.
\end{theorem}
\begin{proof}
Let us define $\Tst$ as in \autoref{def:exp}. By~\autoref{thm:while-kleene},
$(\BC,\Tame)$ is a KiCT iff $\BC$ supports a while-operator%
, satisfying the laws in~\autoref{fig:twhile}. By \autoref{pro:test_while},
the latter is the case iff~$\BC$ supports $\Tame$-uniform Conway 
iteration.
\end{proof}

\begin{ffigure}
\fbox{\parbox{\linewidth}{
\begin{flalign*}
\axname{TW-Fix:}  &\quad   \wt{b}{p} = \ift{b}{p\dc (\wt{b}{p})}{\id}\\[1ex]
\axname{TW-Or:}   &\quad   \wt{(b\lor c)}{p} = (\wt{b}{p})\dc \wt{c}{(p\dc \wt{b}{p})}\qquad\\[1ex]
\axname{TW-Uni:}  &\quad   \frac{u\dc \bar b = \bar c\dc v\qquad u\dc b\dc p = c\dc q\dc u}
                            {u\dc \wt{b}{p} = (\wt{c}{q})\dc v}
\end{flalign*}
}}
  \caption{Uniform Conway iteration in terms of tests. %
  }
  \label{fig:twhile}
\vspace{-2em}
\end{ffigure}
Now we can prove \autoref{the:resmp}.
\begin{proof}[Proof~\autoref{the:resmp} (Sketch)]
We need to check that ${(\BV_{\BBT_H},\Tame)}$
is a KiC. By \autoref{thm:elgot-grove}, we equivalently prove that
$\BV_{\BBT_H}$ supports $\Tame$-uniform Conway iteration.
It is already known~\cite[Lemma 7.2]{GoncharovSchroderEtAl18} that if $\BBT$ supports Conway iteration,  
then so does~$\BBT_H$. 
By \autoref{thm:elgot-grove}, we are left to check that $\BBT_H$ satisfies $\UNI$,
which is a matter of calculation.
\end{proof}

\section{Free KiCTs and Completeness}\label{sec:free}
In this section, we characterize a free KiCT with strict coproducts (i.e.\ those, 
for which coherence maps $X\cpp (Y\cpp Z)\iso (X\cpp Y)\cpp Z$ are identities) on a one-sorted 
signature. We achieve this by combining techniques from formal languages~\cite{Brzozowski64},  
category theory and the theory of Elgot iteration with coalgebraic reasoning~\cite{Rutten03}, in particular proofs by coalgebraic bisimilarity.
We claim that a more general characterization of a free KiCT on a multi-sorted signature can be achieved 
along the same lines, modulo a significant notation overhead and the necessity to form 
final coalgebras in the category of multisorted sets $\Set^\CS$ where $\CS$ is the set of sorts.
We dispense with this option for the sake of brevity and readability.
Let us fix 
\begin{itemize}
  \item a signatures of $n$-ary symbols $\Sigma_n$ for each $n\in\nat$, and let $\Sigma=\bigcup_n\Sigma_n$; 
  \item a signature $\Gamma$ of (unary) symbols, disjoint from $\Sigma$; 
  \item a finite (!) signature $\Theta$ of (unary) symbols, disjoint from $\Sigma\cup\Gamma$.  
\end{itemize}
Let $\hat\Theta$ denote the set of finite subsets of $\Theta$.
We regard $\Theta$ as a signature for tests, $\Gamma$ as a signature for tame morphisms
and~$\Sigma$ as a signature for general morphisms;~$\hat\Theta$ is meant to capture 
finite conjunctions of the form $\bs_1\land\ldots\land\bs_n\land\bar\bs_{n+1}\land\ldots\land\bar\bs_m$
as semantic correspondents of subsets $\{\bs_1,\ldots,\bs_n\}\in\hat\Theta$, assuming an enumeration 
$\Theta = \{\bs_1,\ldots,\bs_m\}$. This is inspired 
by Kleene algebra with tests~\cite{Kozen97a}. Furthermore, we accommodate 
\emph{guarded strings} from \emph{op.~cit.}: let 
$\Gamma^\Theta$
be the set of strings 
	$\brks{b_1,\us_1,\ldots,b_n,\us_n,b_{n+1}}$
with $\us_i\in\Gamma$, $b_i\in \hat\Theta$.
\subsection{Interpretations}
An \emph{interpretation} $\sem{\argument}$ of $(\Sigma,\Gamma,\Theta)$ over a KiCT 
$(\BC,\Tame,\Tst)$ is specified as follows:
\begin{align*}
\sem{1}\in\;& |\BC|&
\sem{\fs}\in\;&\BC(\sem{1},\sem{n})\qquad
(\fs\in\Sigma_{n})\\[1ex]
\sem{\us}\in\;&\Tame(\sem{1},\sem{1})\qquad(\us\in\Gamma)&
\sem{\bs}\in\;&\Tst(\sem{1},\sem{1})\qquad (\bs\in\Theta)%
\end{align*}
where $\sem{n}$ abbreviates the $n$-fold sum $\sem{1}\cpp\ldots\cpp\sem{1}$.
The latter immediately extends to~$\hat\Theta$: $\sem{\{\,\}}=\id$, 
$\sem{\{\bs_1,\ldots,\bs_n\}} = \bs_1\dc\ldots\dc\bs_n\dc\bar\bs_{n+1}\dc\ldots\dc\bar\bs_m$,
assuming that $\Theta = \{\bs_1,\ldots,\bs_m\}$.
Note that we interpret $n$-ary symbols over $\BC(\sem{1},\sem{n})=\BC^{\op}(\sem{1}^n,\sem{1})$.
This equation seems to suggest that it could be more natural to use categories with products as models, 
rather than categories with coproducts. Our present choice helps us to treat generic
KiCTs on the same footing with the free KiCT,
which is defined in terms of coproducts and not products.  
\begin{definition}[Free KiCT]\label{def:free_sg}
A \emph{free KiCT} w.r.t.\ $(\Sigma,\Gamma,\Theta)$ is a KiCT $(\Frep_{\Sigma,\Gamma,\Theta},\overline{\Frep_{\Sigma,\Gamma,\Theta}},{\Frepq_{\Sigma,\Gamma,\Theta}})$  together with an interpretation of $(\Sigma,\Gamma,\Theta)$ in $\Frep_{\Sigma,\Gamma,\Theta}$, such that for any other interpretation of $(\Sigma,\Gamma,\Theta)$ over a KiCT $(\BC,\Tame,\Tst)$, there is unique compatible functor from $\Frep_{\Sigma,\Gamma,\Theta}$
to $\BC$. More formally, for any interpretation $\sem{\argument}$, there is unique KiCT-functor 
$\sem{\argument}_\shortuparrow\c (\Frep_{\Sigma,\Gamma,\Theta},\overline{\Frep_{\Sigma,\Gamma,\Theta}},{\Frepq_{\Sigma,\Gamma,\Theta}})\to(\BC,\Tame,\Tst)$ 
such that the diagram 
\begin{equation}\label{eq:free-cpp}
\begin{tikzcd}[column sep=5em, row sep=normal]
\Frep_{\Sigma,\Gamma,\Theta}
	\rar["\sem{\argument}_\shortuparrow"] & \BC\\
(\Sigma,\Gamma,\Theta)
	\uar["\sem{\argument}^{\Fre}"]
	\ar[ur,"\sem{\argument}"']           
\end{tikzcd}
\end{equation}
commutes.
\end{definition}
In what follows, we characterize $\Frep_{\Sigma,\Gamma,\Theta}$ as a certain category 
of rational trees, i.e.\ trees with finitely many distinct subtrees. An alternative,
equivalent formulation would be to view~$\Frep_{\Sigma,\Gamma,\Theta}$ as a free 
model of the (Lawvere) theory of KiCTs.

Like in the case of original Kozen's completeness result~\cite{Kozen94}, a characterization
of the free model immediately entails completeness of the corresponding axiomatization over it. 
Indeed, by generalities, a free KiCT is isomorphic to the free algebra of 
terms, quotiented by the provable equality relation. Hence, if an equality holds 
over the free model, it is provable.

\subsection{A KiCT of Coalgebraic Resumptions}

For any set~$X$, define
	$TX=\CR (\Gamma^\Theta\times X)$ and
	$T_\nu X=\nu\gamma.\,T(X \cpp \Sigma\gamma)$,
in the category of sets~$\Set$ where $\CR$ is the monad from \autoref{exa:ini}.

A stepping stone for constructing $\Frep_{\Sigma,\Gamma,\Theta}$ is the 
observation that $(\Set_{\BBT_\nu},\Set_{\BBT})$ forms a KiC.
Indeed, $(\Set_{\CR},\Set_{\CR })$ is a KiC and the monad $\CR $ is commutative, 
hence symmetric monoidal. In $\Set_{\CR}$, $\Gamma^\Theta$ is a monoid
under the following operations:
\begin{align*}
\brks{w_1,\ldots,w_{n+1}}\cdot \brks{u_1,\ldots, u_{m+1}} = \begin{cases}
    \brks{w_1,\ldots,w_{n},u_2,\ldots, u_{m+1}}  & \text{if } w_n=u_1\\
    \zero 					& \text{otherwise}
  \end{cases}
\end{align*}
This produces the monad $\BBT$, whose Kleisli category is a KiC by 
\autoref{the:grove-monad}, analogously to \autoref{exa:fin-tra}. Now, 
$(\Set_{\BBT_\nu},\Set_{\BBT})$ is a KiC by \autoref{the:resmp}. We will use the following 
representation for generic elements of $T_\nu X$: 
\begin{align}\label{eq:gen-tree}
	t=\sum_{i\in I} b_i.\,\us_i.\,t_i+\sum_{i\in J} b_i.\,\fs_i(t_{i,1},\ldots,t_{i,n_i}) + \sum_{i\in K} b_i.\,x_i
\end{align}
where $I$, $J$, $K$ are mutually disjoint countable sets, $b_i$ range over $\hat\Theta$, 
$\us_i$ range over $\Gamma$, $\fs_i$ range over $\Sigma$, $t_i,t_{i,j}$ range over $T_\nu X$ 
and $x_i$ range over $X$. 

\begin{definition}[Derivatives]
For every $t\in T_\nu X$, as in~\eqref{eq:gen-tree}, define the following \emph{derivative} operations:
\begin{itemize}
  \item $\partial_{b,\us}(t)=\sum_{i\in I,b=b_i,\us=\us_i} t_i$, for $b\in\hat\Theta$, $\us\in\Gamma$;
  \item $\partial_{b,\fs}^k(t)=\sum_{i\in J,b=b_i,\fs=\fs_i} t_{i,k}$, for $b\in\hat\Theta$, $k\in\{1,\ldots,n_i\}$, $\fs\in\Sigma_{n_i}$ with $n_i>0$.
\end{itemize}
Additionally, let $o(t)= \sum_{i\in K} b_i.\,x_i$.
We extend these operations to arbitrary morphisms $Y\to T_\nu X$ pointwise. 
The \emph{set of derivatives} of $t\in T_\nu X$ is the smallest set $\Der(t)$ that 
contains~$t$ and is closed under all $\partial_{b,\us}$ and $\partial_{b,\fs}^k$. 
\end{definition}
The following property is a direct consequence of these definitions:
\begin{lemma}\label{lem:D-comp}
Let $t\in T_{\nu}X$ be as in~\eqref{eq:gen-tree}, and let $s\c X\to T_\nu Y$.
Then 
\begin{align*}
	\partial_{b,\us}(t\dc s^\klstar) 	=&\; \partial_{b,\us}(t)\dc s^\klstar\plus o(t)\dc (\partial_{b,\us}(s))^\klstar& o(t\dc s^\klstar) =&\; o(t)\dc (o(s))^\klstar\\
	\partial_{b,\fs}^k(t\dc s^\klstar) 	=&\; \partial_{b,\fs}^k(t)\dc s^\klstar\plus o(t)\dc (\partial_{b,\fs}^k(s))^\klstar
\end{align*} 
\end{lemma}
\begin{lemma}\label{lem:bisim}
Given a set $X$, let $\Bis\subseteq T_\nu X\times T_\nu X$ be 
such a relation that whenever~$t\Bis* s$,
\begin{enumerate}
  \item $\partial_{b,\us}(t)\Bis*\partial_{b,\us}(s)$ for all $b\in\hat\Theta$, $\us\in\Gamma$,
  \item $\partial_{b,\fs}^k(t)\Bis*\partial_{b,\fs}^k(s)$ for all $b\in\hat\Theta$, $\fs\in\Sigma$,
  \item $o(t) = o(s)$.
\end{enumerate} 
Then, $t=s$ whenever $t\Bis* s$.
\end{lemma}
\begin{proof}[Proof Sketch]
It suffices to show that $\Bis$ is a coalgebraic bisimulation. The claim is then a 
consequence of strong extensionality of the final coalgebra $T_\nu X$. Let us spell 
out what it means for $\Bis$ to be a coalgebraic bisimulation. Given $t$ and $t'$, such 
that $t\Bis* t'$, and assuming representations
\begin{align*}
t=&\;\sum_{i\in I} 	 g_i.\,\fs_i(t_{i,1},\ldots,t_{i,n_i}) + \sum_{i\in J} g_i.\,x_i,\\*
t'=&\;\sum_{i\in I'} g_i.\,\fs_i(t_{i,1},\ldots,t_{i,n_i}) + \sum_{i\in J'} g_i.\,x_i
\end{align*}
where the $g_i$ range over $\Gamma^\Theta$, the sums $\sum_{i\in J} g_i.\,x_i$
and $\sum_{i\in J'} g_i.\,x_i$ must be equal, and there must exist a set $K$ and 
surjections $e\c K\to I$, $e'\c K\to I'$, such that for every $k\in K$, 
$g_{e(k)} = g_{e'(k)}$, $\fs_{e(k)} = \fs_{e'(k)}$ and $t_{e(k),1}\Bis* t_{e'(k),1},
\ldots,t_{e(k),m}\Bis* t_{e'(k),m}$ where $m$ is the arity of~$\fs_{e(k)}$. This is indeed true 
for $\Bis$. The reason for it is that $t$ and $t'$ can be represented as
\begin{align*}
t=&\;\sum_{n\in\nat}\sum_{i\in I, |g_i|=n} 	 g_i.\,\fs_i(t_{i,1},\ldots,t_{i,n_i}) + \sum_{n\in\nat}\sum_{i\in J, |g_i|=n} g_i.\,x_i,\\*
t'=&\;\sum_{n\in\nat}\sum_{i\in I', |g_i|=n} g_i.\,\fs_i(t_{i,1}',\ldots,t_{i,n_i}') + \sum_{n\in\nat}\sum_{i\in J', |g_i|=n} g_i.\,x_i
\end{align*}
and we can derive the requisite properties for inner sums by induction on $n$
from the assumptions.\end{proof}

\subsection{Rational Trees}
In what follows, we identify every $n\in\nat$ with the set $\{0,\ldots,n-1\}$,
and select binary coproducts in $\Set$ so that $n\oplus m = \{0,\ldots,n-1\}\oplus \{0,\ldots,m-1\} = \{0,\ldots,n+m-1\} = n + m$.
The inclusion of $n$ to $m\geq n$ is then a coproduct injection, which we refer to 
as $\inj_{n}^m$. %
\begin{definition}[Prefinite, Flat, (Non-)Guarded, Rational, Definable]
\begin{enumerate}%
  \item The set of \emph{prefinite} elements of~$T_\nu X$ is defined by induction: 
$t\in T_\nu X$ of the form~\eqref{eq:gen-tree} is \emph{prefinite} if the involved 
sums contain finitely many distinct elements and all the ${t_i,t_{i,j}\in T_\nu X}$ 
are prefinite. 
  \item A prefinite $t\in T_\nu X$ of the form~\eqref{eq:gen-tree} is \emph{flat} if ${t_i,t_{i,j}\in X}$.
\item An element $t\in T_\nu X$ of the form~\eqref{eq:gen-tree} is \emph{guarded} if $K=\emptyset$.
\item An element $t\in T_\nu X$ of the form~\eqref{eq:gen-tree} is \emph{non-guarded} if $I\cup J=\emptyset$.
\item An element $t\in T_\nu X$ %
is \emph{rational} if $\Der(t)$ is finite and $t$ depends on a finite subset~of~$\Sigma\cup\Gamma$.
\end{enumerate}
\noindent A map $t\c Y\to T_\nu X$ is prefinite/flat/guarded/non-guarded if correspondingly 
for every $x\in X$, every $t(x)$ is prefinite/flat/guarded/non-guarded. Finally:
\begin{enumerate}\setcounter{enumi}{5} %
\item A map $t\c k\to T_\nu n$  (with $k,n\in\nat$) is \emph{definable} if for some $m\geq k$ 
there is flat guarded $s\c m\to T_\nu m$ and non-guarded $r\c m\to T_\nu n$, 
such that $t = \inj_k^m\dc s^\rstar\dc r^\klstar$.
\end{enumerate}
\end{definition}
Using \autoref{lem:D-comp}, one can show
\begin{lemma}\label{lem:fix-rat}
Sum, composition and Kleene iteration of rational maps are again rational.
\end{lemma}
The following property is a form of Kleene theorem, originally stating the equivalence of regular 
and recognizable languages~\cite{Sakarovitch09}. In our setting it is proven with the help of \autoref{lem:bisim}.\par
\begin{proposition}\label{pro:kleene}
Given $n,k\in\nat$, a map $n\to T_\nu k$ is rational iff it is definable.
\end{proposition}
Let $\Frep=\Frep_{\Sigma,\Gamma,\Theta}$ be the (non-full) subcategory of  
$\Set_{\BBT_\nu}$, identified as follows:
\begin{itemize}
  \item the objects of $\Frep$ are positive natural numbers,
  \item the morphisms in $\Frep(n,k)$ are rational maps $f\c n\to T_\nu k$ 
  (equivalently: (co)tuples $[t_0,\ldots,t_{n-1}]$ of rational elements of~$T_\nu k$). 
\end{itemize}
Let the wide subcategory of tame morphisms $\overline{\Frep}$ consist of such 
tuples $[t_0,\ldots,t_{n-1}]$ that~${t_i\in Tk}$ for all $i$, and let $\Frepq$ consist of 
those maps in $\overline{\Frep}$ that do not involve symbols from~$\Gamma$.
This defines a KiCT essentially due to the closure properties from \autoref{lem:fix-rat}.

Given an interpretation $\sem{\argument}\c (\Sigma,\Gamma,\Theta)\to\BC$, let us 
extend it to flat elements first via
\begin{gather*}
\hspace{-.2em}\sem{\zero} = \zero,\quad~
\sem{t\plus s} = \sem{t}\plus \sem{s},\quad~
\sem{\eta^\rstar} = \sem{1}^\rstar,\quad~
\sem{t\dc s^\klstar} = \sem{t}\dc\sem{s},\quad~
\sem{k} = \inj_k~(k\in n)
\end{gather*}
where $\eta^\rstar$ stands for a tuple of infinite sum $[\sum_{i\in\nat} \{\}.0,\ldots, \sum_{i\in\nat} \{\}.(n-1)]$,
and is the interpretation of $\eta^\star$ in $\Frep$. The clause for $\eta^\rstar$ 
in necessary to cater for infinite sums that can occur in prefinite elements. Such sums 
can only contain finitely many distinct elements, and thus can be expressed
via finite sums and composition with $\eta^\rstar$. Next, define ${\sem{\argument}_\shortuparrow\c \Frep\to\BC}$
\begin{itemize}
  \item on objects via $\sem{n}_\shortuparrow = \sem{1}\cpp\ldots\cpp \sem{1}$ ($\sem{1}$ repeated $n$ times),
  \item on morphisms, via $\sem{t}_\shortuparrow= \inj_{n,m}\dc \sem{s}^\rstar\dc \sem{r}$, where 
  $t=\inj_{n,m}\dc t^\rstar\dc r^\klstar$, for a guarded flat $s\c m\to T_\nu m$, and a non-guarded $r\c m\to T_\nu k$, 
  computed with \autoref{pro:kleene}.
\end{itemize}
\begin{theorem}\label{the:free}
$\Frep_{\Sigma,\Gamma,\Theta}$ is a free KiCT over $\Sigma,\Gamma,\Theta$.
\end{theorem}
The following property is instrumental for proving this result:
\begin{lemma}\label{lem:funct}
Let $(\BC,\Tame)$ and $(\BD,\Tame[\BD])$ be two KiCs, and let $F$ be the following 
map, acting on objects and on morphisms: $FX\in |\BD|$ for every $X\in |\BC|$, 
$Fp\in\BD(FX,FY)$ for every $p\in\BC(X,Y)$. Suppose that $F$ preserves coproducts, $Fp\in\Tame[\BD](FX,FY)$ for all $p\in\Tame(X,Y)$.
$F$ is a KiC-functor if the following further preservation properties hold
\begin{align*}
Fp^\rstar = (Fp)^\rstar,\quad F(p\dc\inl) = Fp\dc\inl,\quad F(p\dc\inr) = Fp\dc\inr,\quad F(p\dc [\zero,\id]) = Fp\dc [\zero,\id].
\end{align*}
\end{lemma}
Let us briefly outline a potential application of \autoref{the:free} to may-diverge
Kleene algebras, which we informally described in the introduction. Let us now define them
formally:
\begin{definition}[May-Diverge Kleene Algebra]
A \emph{may-diverge Kleene algebra} is an idempotent semiring $(S\comma\zero\comma\id\comma\plus\comma\,\dc)$ 
equipped with an iteration operator $(\argument)^\rstar\c S\to S$ satisfying the laws: 
\begin{gather*}
p^\rstar = \id \plus p\dc p^\rstar
\qquad
(p \plus  q)^\rstar =  p^\rstar\dc (q\dc p^\rstar)^\rstar
\qquad
\vcenter{\infer{r\dc p^\rstar = q^\rstar\dc r}{r\dc p = q\dc r}}
\end{gather*}
\end{definition}
Thus, may-diverge Kleene algebras are very close to KiCs of the form $(\BC,\BC)$ with $|\BC|=1$,
except that in our present treatment all KiCs come with binary coproducts as an 
additional structure. We conjecture though that any may-diverge Kleene algebra,
viewed as a category, can be embedded to a KiC $(\BC,\BC)$ with $|\BC|=\{1,2,\ldots\}$. \autoref{the:free}
will then entail a characterization of the free may-diverge Kleene algebra on $\Gamma$ 
as the full subcategory induced by the single object $1$ of the Kleisli category of the monad $TX=\CR (\Gamma^\star\times X)$. In other words, the free may-diverge Kleene algebra is carried (up-to-isomorphism) 
by rational elements of $\CR(\Gamma^\star)$, similarly to that how the free Kleene 
algebra is carried by rational elements of $\PSet(\Gamma^\star)$.

\section{Conclusions and Further Work}\label{sec:conc}
We developed a general and robust categorical notion of Kleene iteration -- 
KiC(T) (=Kleene-iteration category (with tests)) -- inspired 
by Kleene algebra (with tests) and its numerous cousins. We attested this 
notion with various yardsticks: stability under the generalized 
coalgebraic resumption monad transformer (hence under the exception transformer,
as its degenerate case), equivalence to the classical notion of Conway 
iteration and to a suitably axiomatized theory of while-loops, but most remarkably, we 
established an explicit description of the ensuing free model, as a category of certain 
nondeterministic rational trees, playing the same role for our theory as the 
algebra of regular events for Kleene algebra. However, in our case, the free model is 
much more intricate and difficult to construct, as the iteration operator of it
is neither a least fixpoint nor a unique fixpoint. A salient feature of our notion, mirrored 
in the structure of the free model, is that it can mediate between linear time and 
branching time semantics via corresponding specified classes 
of morphisms.

Given the abstract nature of our results, we expect them be be reusable for varying 
and enriching the core notion of Kleene iteration with other features. For example,
our underlying notion of nondeterminism is that of idempotent 
grove category. General grove categories are a natural base for probabilistic or graded semantics,
and we expect that most of our results, including completeness can be adapted to this case.
Yet more generally, a relevant ingredient of our construction is monad $\CR$, currently capturing the effect 
of nondeterminism, but which can potentially be varied to obtain other flavors 
of linear behavior.

An important open problem that remains for future work is that of defining KiCTs without
coproducts, potentially providing a bridge to relevant algebraic structures as single-object 
categories. Now that the free KiCT with coproduct is identified, the free KiCT 
without coproducts is expected to be complete over the same model. Identifying such a notion 
is hard, because it would simultaneously encompass independent axiomatizations of iterative behavior,
e.g.\ branching time and linear time. As of now, such axiomatizations 
are built on hard-to-reconcile approaches to iteration as either a least or 
a unique fixpoint.

\clearpage

\bibliographystyle{plain}
\bibliography{monads}

\clearpage
\appendix

\section*{Appendix: Selected Proof Details}
\addcontentsline{toc}{section}{Appendices}
\renewcommand{\thesubsection}{\Roman{subsection}.}

\pacman{
\section{Case Study: May-Diverge Kleene Algebra} 
Recall may-diverge Kleene algebras from the introduction. We now can define them
precisely. \par
\begin{definition}[May-Diverge Kleene algebra]
A \emph{may-diverge Kleene algebra} is a set, equipped with the following operations: 
$(\zero/0,\id/0,\dc/2,\plus/2,\rstar/1)$, where $op/i$ indicates that the arity of 
$op$ is $i$, such that the following properties are satisfied: 
\begin{gather}
\quad p\plus p = p \qquad
p\plus q = q\plus p  \qquad
p\plus \zero = p  \qquad
p\plus (q\plus r) = (p\plus q)\plus r\notag
\\[1ex] 
\quad p\dc\id = p  \qquad
\id\dc p = p  \qquad
p\dc (q\dc r) = (p\dc q)\dc r\qquad
\zero\dc p = \zero  \qquad
p\dc\zero = \zero  \qquad\notag
\\[1ex] 
\quad p\dc (q\plus r) = p\dc q\plus q\dc r 
\qquad
(p\plus q)\dc r = p\dc r\plus q\dc r\notag
\\[3ex]
\quad \label{eq:md-kleene}
p^\rstar = \id \plus p\dc p^\rstar
\qquad
\qquad
(p \plus  q)^\rstar = p^\rstar\dc (q\dc p^\rstar)^\rstar
\qquad
\qquad
\raisebox{-.3\baselineskip}{$\infer{r\dc p^\rstar = q^\rstar\dc r}{r\dc p = q\dc r}$}
\end{gather}
\end{definition}
More generally, let us define 
\begin{definition}[May-Diverge Kleene-Kozen Category]\label{def:may-div}
A may-diverge Kleene-Kozen category is an idempotent grove category $\BC$ with all 
morphisms linear, and equipped with 
a Kleene iteration operator 
\begin{displaymath}
  (\argument)^\rstar\c\BC(X,X)\to\BC(X,X),
\end{displaymath}
that satisfies~\eqref{eq:md-kleene}.
\end{definition}
Obviously, a may-diverge Kleene algebra is a may-diverge Kleene-Kozen category 
on one object. We intend to use the above definition as an illustration to our previous developments. 
One may wonder, if the proposed axiomatization is a right one. A justification 
can be an explicit description of the free model, which would appear to be 
convincing. We will show how such a description flows from \autoref{the:free}.

Let $\BC$ be an idempotent grove category in which all morphisms are linear. We 
introduce a category $\Mtx(\BC)$ as follows:
\begin{itemize}
  \item the objects of $\Mtx(\BC)$ are non-empty lists $\brks{A_1,\ldots,A_n}$ of objects of $\BC$;
  \item a morphism $p\c\brks{A_1,\ldots,A_n}\to\brks{B_1,\ldots,B_m}$ in $\Mtx(\BC)$ is given by a family $\brks{p_{i,j}\c A_i\to B_j}_{i\leq n,j\leq m}$
  of morphisms in $\BC$;
  \item the identity morphism over $\brks{A_1,\ldots,A_n}$ is the family $\brks{\delta_{i,j}}_{i,j\leq n}$
  where ${\delta_{i,i} = \id}$ and $\delta_{i,j} = \zero$ if $i\neq j$;
  \item morphism composition in $\Mtx(\BC)$ is calculated with the same formula as
  matrix composition: given $p\c\brks{A_1,\ldots,A_n}\to\brks{B_1,\ldots,B_m}$ and 
  $q\c\brks{B_1,\ldots,B_m}\to\brks{C_1,\ldots,C_k}$, 
  \begin{align*}
  (p\dc q)_{i,j} = p_{i,1}\dc q_{1,j}\plus\ldots\plus p_{i,m}\dc q_{m,j}.
  \end{align*}
\end{itemize} 
The following properties are easy to verify:
\begin{proposition}
Given an idempotent grove category $\BC$ with all morphisms linear,
\begin{enumerate}
  \item $\Mtx(\BC)$ is an idempotent grove category too with $\zero$ and $\plus$
  defined pointwise;
  \item $\Mtx(\BC)$ has strict binary coproducts: $\brks{A_1,\ldots,A_n}\cpp \brks{B_1,\ldots,B_m} = \brks{A_1,\ldots,A_n,B_1,\ldots,B_m}$ with coproduct injections, defined 
  as follows: 
  \begin{align*}
  \inl\c&\brks{A_1,\ldots,A_n}\to\brks{A_1,\ldots,A_n,B_1,\ldots,B_m},\\
  \inr\c&\brks{B_1,\ldots,B_m}\to\brks{A_1,\ldots,A_n,B_1,\ldots,B_m}
  \end{align*}
  are as follows:
  $\inl_{i,j} = \id$, $\inr_{i,j} = \zero$ if $j\leq n$ and $\inl_{i,j} = \zero$, $\inr_{i,j} = \id$
  otherwise.
  \item $\BC$ is a full subcategory of $\Mtx(\BC)$ under the map that sends 
  objects of $\BC$ to singleton lists.
  \item If $\BC$ has binary coproducts with coproduct injections, these coproducts 
  are isomorphic to those in $\Mtx(\BC)$ under $\brks{\inl,\inr}\c \brks{A,B}\to \brks{A\cpp B}$.
\end{enumerate}
\end{proposition}
\paragraph{Free grove categories} We next dwell on free 
grove categories with binary coproducts. To that end, we~fix 
\begin{itemize}
  \item a set of \emph{sorts} $\CS$; 
  \item a signature $\Sigma$ of symbols, together with the corresponding types
  of the form $A\to\brks{B_1,\ldots,B_n}$ where $A,B_1,\ldots,B_n\in\CS$ and $n>0$.
\end{itemize}
We can interpret this data in any grove category $\BC$ with coproducts
by assigning 
\begin{itemize}
  \item an object $\sem{A}$ of $\BC$ to every $A\in\CS$;
  \item a morphism $\sem{f}\c\sem{A}\to\sem{B_1}\cpp\ldots\cpp\sem{B_n}$ to every 
  $f\c A\to\brks{B_1,\ldots,B_n}$ from $\BC$, where we assume $\cpp$ to be right-associative.
\end{itemize}
This induces a notion of \emph{free} grove category with coproducts $\Frep(\CS,\Sigma)$ 
w.r.t.\ $(\CS,\Sigma)$ as a corresponding category together with an interpretation of $(\CS,\Sigma)$ in $\Frep(\CS,\Sigma)$
that for any other interpretation of $(\CS,\Sigma)$ over a nondeterministic 
category with coproducts $\BD$, there is unique compatible functor from $\Frep(\CS,\Sigma)$
to $\BD$. More formally, for any interpretation $\sem{\argument}_{\BD}$ of~$(\CS,\Sigma)$ 
in an idempotent grove category with binary coproducts $\BD$, there is unique 
functor $\sem{\argument}_\shortuparrow\c \Frep(\CS,\Sigma)\to\BD$, preserving 
binary coproducts, $\zero$ and~$\plus$, such that the diagram 
\begin{equation}\label{eq:free-nd}
\begin{tikzcd}[column sep=4em, row sep=normal]
\Frep(\CS,\Sigma)\rar["\sem{\argument}_\shortuparrow"] & \BD\\
(\CS,\Sigma)\uar["\sem{\argument}"]\ar[ur,"\sem{\argument}_{\BD}"']           
\end{tikzcd}
\end{equation}
commutes. We proceed to develop an explicit construction of $\Frep(\CS,\Sigma)$.
Finally, we proceed with a characterization of a free grove category 
without binary coproducts. In this setting, we restrict type profiles of 
symbols in $\Sigma$ to $A\to B$ with $A,B\in\CS$.
\begin{theorem}\label{thm:free-nd-ncp}
The full subcategory $\Fre_\Sigma$ of $\Frep(\CS,\Sigma)$, obtained by restricting to the 
singleton objects is a free grove category on $(\CS,\Sigma)$.
\end{theorem}
\begin{proof}[Sketch]
For any grove category $\BD$, which interprets $(\CS,\Sigma)$ by $\sem{\argument}_{\BD}$, 
we obtain $\sem{\argument}_{\Mtx(\BD)}$ by composition with the inclusion $\BD\ito\Mtx(\BD)$,
which induces $\sem{\argument}_{\Mtx(\BD)}^\sharp\c\Frep(\CS,\Sigma)\to\Mtx(\BD)$. The composition
of $\sem{\argument}_{\Mtx(\BD)}^\sharp$ with the inclusion of $\Fre_\Sigma$ to $\Frep(\CS,\Sigma)$ 
factors though the inclusion of $\BD$ to $\Mtx(\BD)$. This induces the requisite 
functor from $\Fre_\Sigma$ to $\BD$:
\begin{equation}
\makebox[\displaywidth]{%
\begin{tikzcd}[row sep=normal, row sep=normal, ampersand replacement=\&]
\& \Frep(\CS,\Sigma)\ar[rr,"\sem{\argument}_{\Mtx(\BD)}^\sharp"] \& \& \Mtx(\BD)\\
\& \Fre_\Sigma\uar[hook]\rar[dashed]  \& \BD\ar[ur,hook]\\
\& (\CS,\Sigma)\uar["\sem{\argument}"]\ar[ur,"\sem{\argument}_{\BD}"']  \& 
\end{tikzcd}
} \tag*{\qed}
\end{equation}
\noqed
\end{proof}
The key fact that enables the proof of \autoref{thm:free-nd-ncp} is that $\Mtx(\BD)$
faithfully extends a given grove category to an idempotent grove category
with binary coproducts.

}

\subsection{Proof of \autoref{pro:if_as_join}}
Suppose that $d$ is linear. Then
\begin{flalign*}
&& \ifd{d}{p}{q}  &\;= d\dc [q,p]&\\*
&&  &\;= d\dc ([q,\zero]\plus [\zero,p]) \\
&&  &\;= d\dc[q,\zero]\plus d\dc[\zero,p] \\
&&  &\;= d\dc [\id,\zero]\dc q\plus d\dc[\zero,\id]\dc p\\*
&&  &\;= \ift{d?}{p}{q}.&\text{\qed}
\end{flalign*}

\subsection{Proof of \autoref{pro:test_char}}
The necessity is obvious. Let us show sufficiency. For every $b\in\Tst(X)$,
let us fix some choice of $\bar b\in\Tst(X)$, for which the pair $(b,\bar b)$ satisfies contradiction and excluded middle, 
and show that $\Tst(X)$ forms a Boolean algebra, i.e.~$\Tst(X)$
is a complemented distributive lattice. Complementation precisely amounts to the
assumed identities, and we are left to show the laws of distributive lattices. 
Since complements are uniquely defined in Boolean algebras, it will follow that 
$\bar b$ is uniquely determined by $b$. Of course, this is not used in the subsequent proof.

It follows 
by definition that $(\Tst(X),\zero,\plus)$ and $(\Tst(X),\id,\dc)$ are monoids.
Showing that they are idempotent and commutative (hence, are semilattices) amounts to showing that 
$b\dc b=b$ and  $b\dc c = c\dc b$ for all $b, c\in\Tst(X)$. The first identity is shown
as follows, using linearity and the assumed identities: 
\begin{align*}
b 
=&\; b\dc\id\\
=&\; b\dc (b\plus \bar b)\\
=&\; b\dc b\plus b\dc \bar b\\
=&\; b\dc b\plus\zero\\
=&\; b\dc b.
\intertext{
For the second one, note that $1 = b\plus\bar b = b\plus b\plus\bar b = b\plus 1$,
and then
}
b\dc c 
=&\; b\dc c\dc \id\\*
=&\; b\dc c\dc (b\plus\id)\\
=&\; b\dc c\dc b\plus b\dc c\\
=&\; b\dc c\dc b\plus b\dc c\dc b\dc c\\
=&\; b\dc c\dc b\dc(\id\plus c)\\
=&\; b\dc c\dc b,
\end{align*}
where we used the instance of idempotence $b\dc c = b\dc c\dc b\dc c$ that we just 
established. 
Analogously, $c\dc b = b\dc c\dc b$, and hence $b\dc c = c\dc b$. 

Finally, 
distributivity amounts to $(a\plus b)\dc c = a\dc c\plus b\dc c$ and 
$a\plus b\dc c = (a\plus b)\dc (a\plus c)$ for all $a,b,c\in\Tst(X)$.
The first identity is an axiom of idempotent grove categories. The second one is obtained as 
follows:
\begin{align*}
(a\plus b)\dc (a\plus c) 
=&\; a\dc a\plus a\dc c\plus b\dc a\plus b\dc c\qquad\\*
=&\; a\plus a\dc c\plus a\dc b\plus b\dc c\\
=&\; a\dc(\id\plus c\plus b)\plus b\dc c\\
=&\; a\dc\id\plus b\dc c\\
=&\; a\plus b\dc c\tag*{\qed}
\end{align*}

\subsection{Proof of \autoref{lem:exp}}
Negation is defined as follows: 
\begin{align*}
\bar\zero	=&\;\id,&\overline{[\inl,\zero]}=&\;[\zero,\inr],&\overline{b\plus c} =&\; \bar b\dc\bar c,\\
\bar\id		=&\;\zero,& \overline{[\zero,\inr]}=&\;[\inl,\zero],&\overline{b\dc c} =&\; \bar b\plus \bar c. 
\end{align*}
For every $X$, let $\Tst(X)$ be the smallest subset of $\BC(X,X)$ that 
contains $\zero$, $\id$, $[\inl,\zero]$ and $[\zero,\inr]$, and closed under $\plus$ and $\!\dc$. 
By \autoref{pro:test_char}, we need to show that 
every $b\in\Tst(X)$ is linear and satisfies $b\dc \bar b = \zero$, 
$b\plus \bar b = \id$, which we do by induction. Let us strengthen the induction invariant 
by also adding $\bar b\dc b = \zero$. Note that the above equations do not uniquely 
define complement, e.g.\ $\overline{[\inl,\zero]}=\;[\zero,\inr]$ refers to a 
particular decomposition of~$X$ as $X_1\cpp X_2$, while another decomposition 
could theoretically produce a different result. Thus, more precisely, we use the 
fact that every element of $\Tst(X)$ has a representation in the free algebra of 
terms over $\zero$, $\id$, $[\inl,\zero]$, $[\zero,\inr]$, $\plus$ and $\dc$, and 
we run induction over this representation.

\emph{Induction Base.} If $b=\zero$ or $b=\id$, the verification is trivial. 
Consider $b=[\inl,\zero]$. Then 
\begin{align*}
b\dc\zero 
=&\; [\inl,\zero]\dc\zero\\*
=&\; [\inl\dc\zero,\zero]\\
=&\;\zero,\\[1ex]
b\dc (p\plus q) 
=&\; [\inl,\zero]\dc (p\plus q)\\*
=&\; [\inl\dc p\plus\inl\dc q,\zero]\\ 
=&\; [\inl\dc p,\zero]\plus [\inl\dc q,\zero]\\
=&\; b\dc p\plus b\dc q,
\intertext{for any suitably 
typed~$p$ and $q$; analogously $\bar b\dc\zero=\zero$, $\bar b\dc (p\plus q) = \bar b\dc p\plus \bar b\dc q$. Moreover,}
b\dc \bar b =&\; [\inl,\zero]\dc [\zero,\inr]\\
 =&\; [\inl\dc [\zero,\inr],\zero\dc [\zero,\inr]]\\
 =&\; \zero,\\[1ex]
b\plus \bar b
 =&\; [\inl,\zero]\plus [\zero,\inr] \\
 =&\; [\inl\plus\zero,\zero\plus\inr] \\
 =&\; \id.
\end{align*}
The identity $\bar b\dc b = \zero$ is shown analogously to $b\dc \bar b=\zero$.

\emph{Induction Step.} 
Suppose that the induction hypothesis holds for some ${b,c\in\Tst(X)}$ 
and show that $b\dc c$, $b\plus c$ are both linear and satisfy equations:
\begin{align}
b\dc c\dc (\bar b\plus\bar c) 	=&\; \zero,& (\bar b\plus\bar c)\dc b\dc c  =&\; \zero,& b\dc c\plus (\bar b\plus\bar c) =&\; \id,\label{eq:exp-tests1}\\*
(b\plus c)\dc \bar b\dc \bar c  =&\; \zero,& \bar b\dc \bar c\dc (b\plus c) =&\; \zero,& (b\plus c)\plus \bar b\dc \bar c =&\; \id.\label{eq:exp-tests2}
\end{align}
\begin{align*}
\intertext{\emph{Linearity of $(b\dc c)$:}} 
b\dc c\dc\zero 
&\;= b\dc\zero\\*
&\;= \zero,
\\[1ex]
b\dc c\dc (p\plus q) 
&\;= b\dc (b\dc p\plus c\dc q)\\
&\;= b\dc c\dc p\plus b\dc b\dc q.
\intertext{\emph{Linearity of $(b\plus c)$:}}
(b\plus c)\dc\zero 
&\;= b\dc\zero\plus c\dc\zero\\
&\;= \zero,
\\[1ex]
(b\plus c)\dc (p\plus q)
&\;= b\dc (p\plus q)\plus c\dc (p\plus q)\\
&\;= (b\dc p\plus b\dc q)\plus (c\dc p\plus c\dc q)\\
&\;= (b\plus c)\dc p\plus (b\plus c)\dc q.
\intertext{Finally, we prove two equations from the list~\eqref{eq:exp-tests1}--\eqref{eq:exp-tests2} --
the rest is analogous. Note that $c\plus\id = c\plus c\plus\bar c = c\plus\bar c=\id$,}
b\dc c\dc (\bar b\plus \bar c) 
&\;= b\dc c\dc \bar b\plus b\dc c\dc \bar c\\
&\;= b\dc c\dc \bar b\plus b\dc\zero \\
&\;= b\dc c\dc \bar b\\
&\;= b\dc c\dc \bar b \plus b\dc \bar b\\
&\;= b\dc (c\plus\id)\dc \bar b \\
&\;= b\dc \id\dc \bar b \\
&\;= \zero.\\[1ex]
b\dc c\plus (\bar b\plus \bar c) 
&\;= b\dc c\plus \bar b\dc\id \plus \id\dc \bar c\\ 
&\;= b\dc c\plus \bar b\dc (c\plus \bar c) \plus (b\plus \bar b)\dc \bar c\\ 
&\;= b\dc c\plus \bar b\dc c\plus \bar b\dc \bar c\plus b\dc\bar c\plus \bar b\dc \bar c\\ 
&\;= b\dc (c\plus \bar c)\plus \bar b\dc (c\plus\bar c)\\
&\;= (b\plus\bar b)\dc (c\plus \bar c)\\
&\;= \id\dc\id\\
&\;=\id.\tag*{\qed}
\end{align*}

\subsection{Proof of \autoref{lem:tests-dec}}
Clause~1 is shown as follows:
\begin{flalign*}
&& (\dec b)? &\;= (\bar b\dc\inl \plus b\dc\inr)\dc [\zero,\id] &\\
&&  &\;= \bar b\dc\inl\dc [\zero,\id]\plus b\dc\inr\dc [\zero,\id]  &\\
&&  &\;= \bar b\dc\zero\plus b &\\
&&  &\;= b.
\intertext{
We proceed with Clause~2. To that end, fix some $d$ from the image
of~$\dec$, i.e.\ assume that $d= \dec b=\bar b\dc\inl \plus b\dc\inr $ for some $b\in\Tst(X)$.
Then
}
&& (\bar b\dc\inl \plus b\dc\inr)\dc\zero &\;= \bar b\dc\inl\dc\zero \plus b\dc\inr\dc\zero  &\\
&&  &\;=  \bar b\dc\zero \plus  b\dc\zero  &\\
&&  &\;=  \zero, &\\[2ex]
&& (\bar b\dc\inl \plus b\dc\inr)\dc(g\plus h) &\;= \bar b\dc\inl\dc (g\plus h)\plus b\dc\inr\dc(g\plus h) &\\
&&  &\;=  \bar b\dc\inl\dc g \plus\bar b\dc\inl\dc h \plus b\dc\inr\dc g\plus b\dc\inr\dc h &\\
&&  &\;=  (\bar b\dc\inl \plus b\dc\inr)\dc g\plus (\bar b\dc\inl \plus b\dc\inr)\dc h, &\\[2ex]
&& (\bar b\dc\inl\plus b\dc\inr)\dc\nabla  &\;= b\dc\inl\dc\nabla\plus\bar b\dc\inr\dc\nabla  &\\
&&  &\;= b\plus\bar b  &\\
&&  &\;= \id,&\\[2ex]
&& (\bar b\dc\inl\plus b\dc\inr)&\dc[\inl,\bar b\dc\inl\plus b\dc\inr] \\*
&&  &\;=  \bar b\dc\inl\plus b\dc\bar b\dc\inl\plus b\dc b\dc\inr &\\*
&&  &\;=  \bar b\dc\inl\plus b\dc\inr,
\intertext{as desired. The dual equation $d=d\dc[d,\inr]$ is shown analogously.
Finally, Clause~3 is shown as follows:}
&& ( e\dc [d,\inr])?  &\;= (e\dc [d,\inr])\dc [\zero,\id] &\\*
&&  &\;= e\dc[d\dc[\zero,\id],\id]  &\\
&&  &\;= e\dc[d\dc[\zero,\id]\plus\zero, d\dc[\zero,\id]\plus\id]   &\\
&&  &\;= e\dc[d\dc[\zero,\id],d\dc[\zero,\id]]\plus e\dc[\zero,\id] &\\
&&  &\;= d\dc[\zero,\id]\plus e\dc[\zero,\id]  &\\
&&  &\;= d?\plus e?.  &\\[2ex]
&& (e\dc[\inl,d])?  &\;= (e\dc[\inl,d])\dc [\zero,\id]  &\\
&&  &\;= e\dc[\zero,d\dc[\zero,\id]]                      &\\
&&  &\;= e\dc[\zero,\id]\dc d\dc[\zero,\id] &\\*
&&  &\;= e?\dc d?.  & %
\intertext{
Finally, for Clause~4, let $d=\bar b\dc\inl \plus b\dc\inr $. Then}
&& (d\dc[\inr,\inl])? &\;= (\bar b\dc\inl \plus b\dc\inr)\dc [\inr,\inl]\dc [\zero,\id] &\\
&&  &\;= \bar b\plus b\dc\zero  &\\
&&  &\;= \bar b &\\
&&  &\;= \overline{\bar b\dc\zero \plus b} &\\
&&  &\;= \overline{(\bar b\dc\inl \plus b\dc\inr)\dc[\zero,\id]} &\\
&&  &\;= \overline{d?}&\tag*{\qed}
\end{flalign*}

\subsection{Proof to \autoref{exa:sqr}}
We will prove the following auxiliary identity: 
\begin{align}\label{eq:sq_aux}
[p\dc\inr,p\dc\inl\plus p\dc\inr]^\rstar = (\id\plus [\zero,p\dc\inl])\dc ([p,p]\dc (\inr\plus p\dc\inl))^\rstar,
\end{align}
from which the goal follows. Indeed, note the identities: 
\begin{align*}
[p\dc\inr,p\dc\inl\plus p\dc\inr]\dc\nabla =&\; \nabla\dc p,\\*
[p,p]\dc (\inr\plus p\dc\inl)\dc\nabla =&\; \nabla\dc p\dc (\id\plus p),
\end{align*}
compose both sides of~\eqref{eq:sq_aux} with $\inl$ on the left and $\nabla$ on
the right, and simplify the result; notice that $\nabla = [1;1] \in \Tame(X \oplus X,X)$:
\begin{flalign*}
&& p^\rstar =&\; \inl\dc\nabla\dc p^\rstar\\* 
&& =&\; \inl\dc [p\dc\inr,p\dc\inl\plus p\dc\inr]^\rstar\dc\nabla &\by{\axname{$\rstar$-Uni}}\\ 
&& =&\; \inl\dc (\id\plus [\zero, p\dc\inl])\dc ([p, p]\dc (\inr\plus p\dc\inl))^\rstar\dc\nabla &\by{\eqref{eq:sq_aux}}\\ 
&& =&\; \inl\dc (\id\plus [\zero, p\dc\inl])\dc\nabla\dc (p\dc (\id\plus p))^\rstar &\by{\axname{$\rstar$-Uni}}\\ 
&& =&\; (p\dc (\id\plus p))^\rstar. &
\intertext{
Now, to show~\eqref{eq:sq_aux}, note that} 
&&[\zero,p\dc\inl]^\rstar 
=&\; \id\plus [\zero, p\dc\inl]\dc [\zero, p\dc\inl]^\rstar&\by{\axname{$\rstar$-Fix}}\\ 
&&=&\; \id\plus [\zero, p\dc\inl\dc [\zero, p\dc\inl]^\rstar]\\
&&=&\; \id\plus [\zero, p\dc\inl\dc (\id\plus [\zero, p\dc\inl]\dc [\zero, p\dc\inl]^\rstar)]&\by{\axname{$\rstar$-Fix}}\\
&&=&\; \id\plus [\zero, p\dc\inl].
\intertext{
Then}
&&[p\dc\inr, p\dc\inl&\plus p\dc\inr]^\rstar \\*
&&=&\; ([\zero, p\dc\inl] \plus [p\dc\inr, p\dc\inr])^\rstar\\
&&=&\; [\zero, p\dc\inl]^\rstar\dc ([p\dc\inr, p\dc\inr]\dc [\zero, p\dc\inl]^\rstar)^\rstar&\by{\axname{$\rstar$-Sum}}\\
&&=&\; (\id\plus [\zero, p\dc\inl])\dc ([p, p]\dc\inr\dc (\id\plus [\zero, p\dc\inl]))^\rstar\\
&&=&\; (\id\plus [\zero, p\dc\inl])\dc ([p, p]\dc (\inr\plus p\dc\inl))^\rstar,
\end{flalign*}
and we are done.\qed

\subsection{Proof of \autoref{the:grove-monad}}
Note that the assumptions imply that $T$ is a KiC-functor, in particular, $T0 = 0$ and $T(p\plus q) = Tp\plus Tq$,
by \autoref{lem:funct}. The laws of grove categories are easy to check. Indeed, the equations~\eqref{eq:dist_plus}
interpreted in $\BC_{\BBT}$ turn into
\begin{align*}
\qquad\qquad \zero\dc p^\klstar =&\; \zero,& (q\plus r)\dc p^\klstar =&\; q\dc p^\klstar\plus r\dc p^\klstar, &&
\end{align*}
and clearly follow from assumptions. Linearity of left injections in $\BC_{\BBT}$ reads as
\begin{align*}
\qquad\qquad \inl\dc\eta\dc\zero^\klstar =&\; \zero,& \inl\dc\eta\dc(q\plus r)^\klstar =&\; \inl\dc\eta\dc q^\klstar\plus  \inl\dc\eta\dc r^\klstar, &&
\end{align*} 
and follows from the assumption that $\inl$ is linear in $\BC$ -- analogously for $\inr$.

Let us check that the morphisms from~$\Tame_\BBT$ are all linear. To that end, 
take some $p$ from~$\Tame_\BBT$, and show 
\begin{align*}
\qquad\qquad p\dc\zero^\klstar =&\; \zero,& p\dc(q\plus r)^\klstar =&\; p\dc q^\klstar\plus  p\dc r^\klstar &&
\end{align*}
Using the assumptions, $p\dc\zero^\klstar = p\dc T\zero\dc\mu = p\dc \zero\dc\mu = \zero\dc\mu = \zero$ where
$\mu=\id^\klstar$ is multiplication of $\BBT$. The second equation is shown analogously.

Finally, we proceed to verify the laws of iteration. To that end, we record some 
auxiliary properties:
\begin{align}
(\eta\dc (p^\klstar)^\rstar)^\klstar &= (p^\klstar)^\rstar, \label{eq:grove-monad1}\\
\eta\dc (p^\klstar\plus q^\klstar) &= p\plus q.				\label{eq:grove-monad2}
\end{align}
For the first identity, note that
\begin{align*}
Tp^\klstar\dc\mu= TTp\dc T\mu\dc\mu&= TTp\dc\mu\dc\mu=\mu\dc Tp\dc\mu=\mu\dc p^\klstar.\\
\intertext{By assumption, $\mu=\id^\klstar$ is tame, hence, by \axname{$\rstar$-Uni},}
(Tp^\klstar)^\rstar\dc\mu&= \mu\dc (p^\klstar)^\rstar.\\
\intertext{We then have}
(\eta\dc (p^\klstar)^\rstar)^\klstar
&=  T\eta\dc T(p^\klstar)^\rstar\dc\mu\\
&=  T\eta\dc (Tp^\klstar)^\rstar\dc\mu\\
&=  T\eta\dc\mu\dc (p^\klstar)^\rstar\\
&=  (\eta\dc (p^\klstar)^\rstar)^\klstar.
\intertext{The second identity is shown as follows:}
\eta\dc (p^\klstar\plus q^\klstar) 
&= \eta\dc (Tp\dc\mu\plus Tq\dc\mu)\\
&= \eta\dc (Tp\plus Tq)\dc\mu\\
&= \eta\dc T(p\plus q)\dc\mu\\
&= (p\plus q)\dc\eta\dc\mu\\
&= p\plus q.
\end{align*}
Now, the laws of KiCs are shown as follows:

\axname{$\rstar$-Fix}: We need to show $\eta\dc (p^\klstar)^\rstar = \eta \plus p\dc (\eta\dc (p^\klstar)^\rstar)^\klstar$. Indeed,
\begin{flalign*}
&& \eta\dc (p^\klstar)^\rstar 
&\;= \eta\dc(\id\plus p^\klstar\dc (p^\klstar)^\rstar) &\\
&&  &\;= \eta\dc(\id\plus p^\klstar\dc (\eta\dc (p^\klstar)^\rstar)^\klstar) &\by{\eqref{eq:grove-monad1}}\\
&&  &\;= \eta\dc(\eta^\klstar\plus (p\dc (\eta\dc (p^\klstar)^\rstar)^\klstar)^\klstar) &\\
&&  &\;= \eta \plus p\dc (\eta\dc (p^\klstar)^\rstar)^\klstar.&\by{\eqref{eq:grove-monad2}}
\end{flalign*}

\axname{$\rstar$-Sum}: We need to show $\eta\dc ((p \plus q)^\klstar)^\rstar =  \eta\dc (p^\klstar)^\rstar\dc (\eta\dc(q\dc (\eta\dc (p^\klstar)^\rstar)^\klstar)^\rstar)^\klstar$.
This is essentially obtained from \axname{$\rstar$-Sum} for $(\BC,\Tame)$ by repeated use of \eqref{eq:grove-monad2}:
\begin{flalign*}
&& \eta\dc ((p\plus q)^\klstar)^\rstar &\;= \eta\dc (p^\klstar \plus  q^\klstar)^\rstar &\\
&&  &\;= \eta\dc (p^\klstar)^\rstar\dc (q^\klstar\dc (p^\klstar)^\rstar)^\rstar \\
&&  &\;= \eta\dc (p^\klstar)^\rstar\dc (q^\klstar\dc (\eta\dc (p^\klstar)^\rstar)^\klstar)^\rstar \\
&&  &\;= \eta\dc (p^\klstar)^\rstar\dc ((q\dc (\eta\dc (p^\klstar)^\rstar)^\klstar)^\klstar)^\rstar \\
&&  &\;= \eta\dc (p^\klstar)^\rstar\dc (\eta\dc (q\dc (\eta\dc (p^\klstar)^\rstar)^\klstar)^\rstar)^\klstar.
\end{flalign*}

\axname{$\rstar$-Uni}: Suppose that $u\dc p^\klstar = q\dc u^\klstar$ for some 
tame $u$. We need to show $u\dc (\eta\dc (p^\klstar)^\rstar)^\klstar = \eta\dc (q^\klstar)^\rstar\dc u^\klstar$.
By applying Kleisli lifting to both sides of the assumption, we obtain 
$u^\klstar\dc p^\klstar = q^\klstar\dc u^\klstar$, which by \axname{$\rstar$-Uni}
yields $u^\klstar\dc (p^\klstar)^\rstar = (q^\klstar)^\rstar\dc u^\klstar$. Then
\begin{flalign*}
&& u\dc (\eta\dc (p^\klstar)^\rstar)^\klstar &\;= u\dc (p^\klstar)^\rstar &\by{\eqref{eq:grove-monad2}}\\
&&  &\;= \eta\dc u^\klstar\dc (p^\klstar)^\rstar\\
&&  &\;= \eta\dc (q^\klstar)^\rstar\dc u^\klstar,
\end{flalign*}
and we are done.
\qed

\subsection{Proof of \autoref{thm:while}}
The present construction is an improvement of a previous one~\cite[Theorem 18]{Goncharov23}, which we partially reuse here.
Two important changes are responsible for redesigning the previous argument: 
here we redefine \axname{DW-And}, and make it an axiom (instead of a rule, as before), 
and we restrict the use of \axname{DW-Uni}, by making more conservative assumptions 
about the subcategory~$\BD$.

The requisite mutual translations are as follows:
\begin{align}\label{eq:while-as-it}
\while{d}{p}  =&\; (\ifd{d}{p\dc \tt}{\ff})^\istar,\\\label{eq:it-as-while}
p^\istar      =&\; \inr\dc (\while{(\inl\cpp \inr)}{[\inl,p]})\dc [\id,\div]
\end{align}
where $\delta=\inr^\istar$.

First, we show that the indicated translations are mutually inverse.
\medskip\noindent
(i)~  $(\argument)^\istar\to\underline{\oname{while}}\to (\argument)^\istar$: We need to show 
  that 
  \begin{displaymath}
     \inr\dc (\ifd{(\inl\cpp \inr)}{[\inl,p]\dc \inr}{\inl})^\istar\dc [\id,\div] = p^\istar.
  \end{displaymath}
Indeed,
\begin{flalign*}
&&  \inr\dc (\ifd{&(\inl\cpp \inr)}{[\inl,p]\dc\inr}{\inl})^\istar\dc [\id,\div]     \\*
&& &\;= \inr\dc ((\inl\cpp \inr)\dc[\inl,[\inl,p]\dc\inr])^\istar\dc [\id,\div]       &\\
&& &\;= \inr\dc (\inl\cpp p)^\istar\dc [\id,\div]                                     &\\
&& &\;= \inr\dc (\id\cpp p)^\istar                                                    &\by{\axname{Naturality}}\\
&& &\;= \inr\dc ([\inl,\inr]\dc[\inl,p\dc\inr])^\istar                                &\\
&& &\;= \inr\dc [\inl,\inr]\dc [\id, (p\dc\inr\dc [\inl,[\inl,\inr]])^\istar]         &\by{\axname{Dinaturality}}\\*
&& &\;= p^\istar.
\end{flalign*}
\medskip\noindent
(ii)~  $\underline{\oname{while}}\to(\argument)^\istar\to\underline{\oname{while}}$: We need to show that
\begin{align}\label{eq:wh-it-wh}
   \inr\dc (\while{(\inl\cpp \inr)}{[\inl,\ifd{d}{p\dc \tt}{\ff}]})\dc [\id,\div] = \while{d}{p}.
\end{align}
Observe that 
\begin{align*}
[\inl,\ifd{d}{p\dc &\tt}{\ff}]\\*
=&\; [\inl, d]\dc (\id\cpp p)\\
=&\; [\inll, d\dc (\inl\cpp \inr)]\dc [\id\cpp p,\id\cpp p]\\*
=&\; \ifd{[\inll, d\dc (\inl\cpp \inr)]}{(\id\cpp p)}{(\id\cpp p)},
\intertext{and}
(\inl\cpp \inr)\booland (&[\inll, d\dc (\inl\cpp \inr)]\boolor\tt)\\
=&\; (\inl\cpp \inr)\dc [\inl, [\inlr,  d\dc\inr]] \\
=&\; [\inll,  d\dc\inr]\\
=&\; \inl \cpp  d,
\end{align*}
Hence, by \axname{DW-And},
\begin{align*}
 \while{(\inl\cpp &\inr)}{[\inl,\ifd{d}{p\dc \tt}{\ff}]} = \while{(\inl\cpp d)}{(\id\cpp p)}.
\end{align*}
We thus reduced~\eqref{eq:wh-it-wh} to
\begin{align}\label{eq:eq:wh-it-wh1}
   \inr\dc \while{(\inl\cpp d)}{(\id\cpp p)}\dc [\id,\div] = \while{d}{p}.
\end{align}
Let $e=[\inll,d\dc(\inr\cpp \inr)]$, and observe that
\begin{align*}
\inl\cpp d =&\; [\inll,d\dc (\inr\cpp \inr)]\dc [\inl\cpp \inl,\inr] \\ 
 =&\; [\inll,d\dc (\inr\cpp \inr)]\boolor (\inl\cpp \inl)\\
 =&\; e\boolor (\inl\cpp \inl).
\intertext{Therefore, by \axname{DW-Or},}
\while{(\inl\cpp d)}{(\id\cpp p)} =&\; (\while{e}{(\id\cpp p)})\dc \\&\;\while{(\inl\cpp \inl)}{((\id\cpp p)\dc \while{e}{(\id\cpp p)})}
\end{align*}
Using the last identity, to show~\eqref{eq:eq:wh-it-wh1}, it suffices to show
\begin{align*}
\inr\dc \while{e}{(\id\cpp p)} =&\; \while{d}{p}\dc \inr\\*
\inr\dc \while{(\inl\cpp \inl)}{((\id\cpp p)\dc \while{e}{(\id\cpp p)})} =&\; \inl
\end{align*}
The latter is straightforward by~\axname{Fixpoint}. For the former, observe that
\begin{align*}
\inr\dc e
=&\;\inr\dc [\inll,d\dc(\inr\cpp \inr)] 
= \ifd{d}{\inr\dc \tt}{\inr\dc \ff},\\*
\inr\dc (\id\cpp p) =&\; p\dc \inr,
\end{align*}
and use \axname{DW-Uni}.

\medskip\noindent
(iii)~  We need to check that the laws of uniform Conway iteration follow from \axname{DW-Fix},
 \axname{DW-Or}, \axname{DW-And}, \axname{DW-Uni}. We use the fact that \axname{Dinaturality}
 is derivable from \axname{Fixpoint},\axname{Codiagonal}, the instance of \axname{Naturality}
 with $q=\id\cpp \inl$ and the instances of \axname{Uniformity} with $u=\inl$ and $u=\inr$~\cite[Proposition 5.1]{GoncharovSchroderEtAl19}.
 Let us argue moreover that \NAT is derivable from the same data, and additionally
 its own instance with $q=\inr$.
To that end, let us fix $q\c Y \to Z$, $p\c X\to Y\cpp X$ and show that
$p^{\istar}\dc q = (p\dc (q\cpp \id))^{\istar}$, by establishing two identities:
\begin{align*}
p^{\istar}\dc q = \inr\dc (q\cpp p)^\istar =&\; (p\dc (q\cpp \id))^{\istar},
\end{align*}
which are obtained as follows, using the allowed instances of \NAT and \UNI:
\begin{flalign*}
&&\inr\dc (q\cpp p)^\istar
  =\; & \inr\dc ([q\dc\inll,p\dc (\inlr\cpp \inr)]\dc [\id,\inr])^\istar                  \\*
&&=\; & \inr\dc [q\dc\inll, p\dc (\inlr\cpp\inr)]^{\istar\istar}                          &\by{\COD}\\
&&=\; & \inr\dc [q\dc\inll,p\dc(\inlr\cpp \inr)]^{\istar}\dc [\id, (q\cpp p)^\istar]      &\by{\FIX}\\
&&=\; & (p\dc (\inlr\cpp\id))^{\istar}\dc [\id, (q\cpp p)^\istar]                         &\by{\UNI}\\
&&=\; & (p\dc (\inr\cpp \id))^\istar \dc (\id\cpp \inl)\dc [\id, (q\cpp p)^\istar]      &\by{\NAT}\\
&&=\; & p^\istar\dc\inr\dc (\id\cpp \inl)\dc [\id, (q\cpp p)^\istar]                      &\by{\NAT}\\
&&=\; & p^\istar\dc\inl\dc (q\cpp p)^\istar                                                 \\
&&=\; & p^\istar\dc\inl\dc (q\cpp p)\dc [\id,(q\cpp p)^\istar]                            &\by{\FIX}\\
&&=\; & p^\istar\dc q\\[1.5ex]
&&\inr\dc (q\cpp p)^{\istar}
=\; & \inr\dc ((q\cpp \id)\dc [\inl,p\dc\inr])^{\istar}\\
&&=\; & \inr\dc (q\cpp \id)\dc [\id,([\inl, p\dc\inr\dc (q\cpp \id)])^{\istar}]&\by{\DIN}\\
&&=\; & \inr\dc[q,(p\dc (q\cpp \id))^{\istar}]\\
&&=\; & (p\dc (q\cpp \id))^{\istar}.
\end{flalign*}
That \FIX and $\COD$ and the above mentioned relevant instances of \UNI and \NAT 
indeed follow has already been shown previously~\cite[Theorem~18]{Goncharov23}

\noindent
(iv)~ That \axname{DW-Fix},\axname{DW-Or},\axname{DW-Uni} follow from the axioms 
of uniform Conway iteration has been proven previously~\cite[Theorem 18]{Goncharov23}. Let us
show \axname{DW-And}:
\begin{flalign*}
&&\while{d\booland (e\boolor\tt)}{p}=&\;  (\ifd{d\booland (e\boolor\tt)}{p\dc \tt}{\ff})^\istar&\\
&&=&\;  (d\dc [\inl,e\dc\nabla\dc\inr]\dc [\inl,p\dc\inr])^\istar\\
&&=&\;  (d\dc (\id\cpp e\dc\nabla\dc p))^\istar\\
&&=&\;  (d\dc [\inl, e\dc [p\dc\inr, p\dc\inr]])^\istar\\
&&=&\;  (\ifd{d}{(\ifd{e}{p}{p})\dc \tt}{\ff})^\istar\\
&&=&\; \while{d}{(\ifd{e}{p}{p})}.&\text{\qed}
\end{flalign*}

\subsection{Proof of \autoref{thm:while-kleene}}
\begin{lemma}\label{lem:while-eqs}
Let $\BC$, $\BD$ and $\Tst$ be as in \autoref{pro:test_while}.
Then the following identities are derivable:
\begin{align}
\label{eq:test_while}
\bar b\dc \wt{(b\land c)}{p} =&\; \bar b\\* 
\label{eq:while_test}
\wt{(b\lor c)}{p} =&\; (\wt{(b\lor c)}{p})\dc \bar b\\  
\label{eq:test_in_while}
\wt{(b\land c)}{p} =&\; \wt{(b\land c)}{(b\dc p)}\\  
\wt{b}{(\ift{c}{p}{q})} =&\;\notag\\ 
(\wt{(b\,\land c)}{p})\dc \,&\wt{b}{(q\dc \wt{(b\land c)}{p})}
\label{eq:wandu}
\end{align}
where the if-then-else operator is defined as in~\eqref{eq:if-while}.
\end{lemma}
\begin{proof}
\eqref{eq:test_while} is shown directly:
\begin{flalign*}
&&\bar b\dc \wt{&(b\land c)}{p}\\*
&& =&\; \bar b\dc \ift{(b\land c)}{p\dc (\wt{(b\land c)}{p})}{\id}&\by{\axname{TW-Fix}}\\
&& =&\; \bar b\dc ((\bar b\plus\bar c)\plus (b\dc c)\dc p\dc \wt{(b\land c)}{p})\\
&& =&\; \bar b.
\end{flalign*}
To show~\eqref{eq:test_while}, note that $\id\dc (\bar b\land\bar c) = (\bar b\land\bar c)\dc \bar b$, 
$\id\dc (b\lor c) = (b\lor c)\dc \id$, $\id\dc p = p\dc \id$, which entail~\eqref{eq:test_while} by~\axname{TW-Uni}.
Equation \eqref{eq:test_in_while} is shown analogously.

Let us proceed with the proof of~\eqref{eq:wandu}. Let us first reduce the 
left-hand side as follows:
\begin{flalign*}
&&\wt{b}{&(\ift{c}{p}{q})} \\*
&&=&\; (\wt{b}{(\ift{c}{p}{q})})\dc \bar b &\by{\eqref{eq:while_test}}\\
&&=&\; (\wt{(b\land c)\lor (b\land\bar c)}{(\ift{c}{p}{q})})\dc \bar b&\\
&&=&\; (\wt{(b\land c)}{(\ift{c}{p}{q})})\dc \\*
&&&\;  (\wt{(b\land\bar c)}{((\ift{c}{p}{q})\dc \\&&&\hspace{7.35em}\wt{(b\land c)}{(\ift{c}{p}{q})})})\dc \bar b&\by{\axname{TW-Or}}\\
&&=&\; (\wt{(b\land c)}{p})\dc\\
&&&\;  \wt{(b\land \bar c)}{(q\dc \wt{(b\land c)}{p})}\dc \bar b.&\by{\eqref{eq:test_in_while}}
\end{flalign*}
Next, observe that $(\bar b\lor\bar c)\dc \bar b = \bar b\dc \bar b$,
$(\bar b\lor\bar c)\dc b = (b\land\bar c)\dc \id$, and hence, by~\axname{TW-Uni}, 
\begin{align*}
(\bar b\lor\bar c)\dc \wt{&b}{(q\dc \wt{(b\land c)}{p})}\\
=&\;
\wt{(b\land\bar c)}{(q\dc \wt{(b\land c)}{p})}\dc \bar b.
\end{align*}
Hence, we can reduce the right-hand side of~\eqref{eq:wandu} to the same expression:
\begin{flalign*}
&&(\wt{(b\,&\land c)}{p})\dc \wt{b}{(q\dc \wt{(b\land c)}{p})} \\*
&&=&\; (\wt{(b\land c)}{p})\dc (\bar b\lor\bar c)\dc \wt{b}{(q\dc \wt{(b\land c)}{p})} &\by{\eqref{eq:while_test}}\\
&&=&\; (\wt{(b\land c)}{p})\dc \wt{(b\land\bar c)}{(q\dc \wt{(b\land c)}{p})}\dc \bar b.&\text{\qed}
\end{flalign*}
\noqed\end{proof}

We use the conventional encoding of while-loops from~\eqref{eq:if-while} together 
with the following one, going in the opposite direction:
\begin{align}\label{eq:rstar_from_wh}
p^\rstar = \inr\dc (\wt{[\zero,\inr]}{[\inl,\inl\plus p\dc \inr]})\dc [\id,\div].
\end{align}
We proceed to show that this yields the desired equivalence. 

\medskip
\noindent
(i)~  $(\argument)^\rstar\to\oname{while}\to (\argument)^\rstar$;  we need to show
\begin{flalign*}
 &&p^\rstar =&\;  \inr\dc ([\zero,\inr]\dc [\inl,\inl\plus p\dc \inr])^\rstar\dc [\inl,\zero]\dc [\id,\div]
 \intertext{Note that, by~\axname{$\rstar$-Fix}, $\inl\dc [\zero, p\dc \inr]^\rstar = \inl$, and hence,}
 &&p^\rstar 
 =&\; p^\rstar\dc \inr\dc [\inl,\inl\plus\inr]\dc [\id,\zero]\\
 &&=&\; p^\rstar\dc \inr\dc [\zero,\inl]^\rstar\dc [\id,\zero]\\
 &&=&\; \inr\dc [\zero, p\dc \inr]^\rstar\dc [\zero,\inl]^\rstar\dc [\id,\zero]&\by{\axname{$\rstar$-Uni}}\\
 &&=&\; \inr\dc [\zero, p\dc \inr]^\rstar\dc [\zero,\inl\dc [\zero, p\dc \inr]^\rstar]^\rstar\dc [\id,\zero]\\
 &&=&\; \inr\dc [\zero, p\dc \inr]^\rstar\dc ([\zero,\inl]\dc [\zero, p\dc \inr]^\rstar)^\rstar\dc [\id,\zero]\\
 &&=&\; \inr\dc ([\zero,\inl]\plus[\zero, p\dc \inr])^\rstar\dc [\id,\zero]&\by{\axname{$\rstar$-Sum}}\\
 &&=&\; \inr\dc [\zero,\inl\plus p\dc \inr]^\rstar\dc [\id,\zero]\\
 &&=&\; \inr\dc ([\zero,\inr]\dc [\inl,\inl\plus p\dc \inr])^\rstar\dc [\inl,\zero]\dc [\id,\div].
\end{flalign*}

\medskip
\noindent
(ii)~  $\oname{while}\to (\argument)^\rstar\to\oname{while}$;  the goal here is
\begin{align*}
\wt{b}{p} = \inr\dc (\wt{[\zero,\inr]}{[\inl,\inl\plus b\dc p\dc \inr]})\dc [\id,\div]\dc \bar b. 
\end{align*}
Note that $(\bar b\plus\id)\dc [\inl,\zero] = [\inl,\zero]\dc (\bar b\plus\id)$, $(\bar b\plus\id)\dc [\zero,\inr] = [\zero,\inr]\dc \id$, 
and $\id\dc [\bar b\dc \inl,\bar b\dc \inl\plus b\dc p\dc \inr] = [\inl,\inl\plus b\dc p\dc \inr]\dc (\bar b\plus\id)$,
and hence, using~\axname{TW-Uni},
\begin{flalign*}
&& \inr\dc (&\wt{[\zero,\inr]}{[\bar b\dc \inl,\bar b\dc \inl\plus b\dc p\dc \inr]})\dc [\id,\div]\dc \bar b&\\
&&  &\;= \inr\dc (\bar b\cpp\id)\dc (\wt{[\zero,\inr]}{[\bar b\dc \inl,\bar b\dc \inl\plus b\dc p\dc \inr]})\dc [\id,\div]\dc \bar b\\
&&  &\;= \inr\dc (\wt{[\zero,\inr]}{[\inl,\inl\plus b\dc p\dc \inr]})\dc (\bar b\plus\id)\dc [\id,\div]\dc \bar b\\
&&  &\;= \inr\dc (\wt{[\zero,\inr]}{[\inl,\inl\plus b\dc p\dc \inr]})\dc [\id,\div]\dc \bar b.
\end{flalign*}
We thus replace the goal with
\begin{align}\label{eq:while_tr}
\wt{b}{p} = \inr\dc (\wt{[\zero,\inr]}{[\bar b\dc \inl,\bar b\dc \inl\plus b\dc p\dc \inr]})\dc [\id,\div]\dc \bar b. 
\end{align}
Next, let $s=[\inr,\inl]$, and observe that
\begin{flalign*}
&&\wt{[\zero,\inr]&}{[\bar b\dc \inl,\bar b\dc \inl\plus b\dc p\dc \inr]}\\
&&=&\; \wt{[\zero,\inr]}{((\zero\cpp\bar b)\dc s\plus (\id\cpp b)\dc (\bar b\cpp p))}\\
&&=&\; \wt{[\zero,\inr]}{(\ift{(\id\cpp b)}{(\bar b\cpp p)}{s})}\\
&&=&\; (\wt{([\zero,\inr]\land (\id\cpp b))}{(\bar b\cpp p)})\dc \\
&&&\; \wt{[\zero,\inr]}{(s\dc \wt{([\zero,\inr]\land (\id\cpp b))}{(\bar b\cpp p)})}&\by{\eqref{eq:wandu}}\\
&&=&\; (\wt{(\zero\cpp b)}{(\bar b\cpp p)})\dc \\*
&&&\; \wt{[\zero,\inr]}{(s\dc \wt{(\zero\cpp b)}{(\bar b\cpp p)})}.&
\end{flalign*}
To finish the proof of~\eqref{eq:while_tr}, it thus suffices to obtain
\begin{align*}
\inr\dc \wt{(\zero\cpp b)}{(\bar b\cpp p)} =&\; (\wt{b}{p})\dc \inr\\
\inr\dc \wt{[\zero,\inr]}{(s\dc \wt{(\zero\cpp b)}{(\bar b\cpp p)})} =&\;\inl
\end{align*}
and then use~\eqref{eq:while_test}. The first identity easily follows from~\axname{TW-Uni}.
The second one follows from~\axname{TW-Fix}.

\medskip
\noindent
(iii) KiCT $\to$ uniform Conway;  let us show \axname{TW-Fix},\axname{TW-Or}
directly:
\begin{flalign*}
&&\wt{b}{&p}\\
&& =&\; (b\dc p)^\rstar\dc \bar b\\
&&=&\; (\id\plus (b\dc p)\dc (b\dc p)^\rstar)\dc \bar b&\by{\axname{$\rstar$-Fix}}\\
&&=&\; \bar b\plus b\dc p\dc (b\dc p)^\rstar\dc \bar b\\
&&=&\; \ift{b}{p\dc (\wt{b}{p})}{\id},\\[2ex]
&&\wt{(b\plus& c)}{p} \\
&&=&\; ((b\plus c)\dc p)^\rstar\dc \bar b\dc \bar c\\
&&=&\; ((b\plus\bar b\dc c)\dc p)^\rstar\dc \bar b\dc \bar c\\
&&=&\; (b\dc p\plus\bar b\dc c\dc p)^\rstar\dc \bar b\dc \bar c\\
&&=&\; (b\dc p)^\rstar\dc (\bar b\dc c\dc p\dc (b\dc p)^\rstar)^\rstar\dc \bar b\dc \bar c&\by{\axname{$\rstar$-Sum}}\\
&&=&\; (b\dc p)^\rstar\dc (\id\plus\bar b\dc c\dc p\dc (b\dc p)^\rstar\dc (\bar b\dc c\dc p\dc (b\dc p)^\rstar)^\rstar)\dc \bar b\dc \bar c&\by{\axname{$\rstar$-Fix}}\\
&&=&\; (b\dc p)^\rstar\dc \bar b\dc (\id\plus c\dc p\dc (b\dc p)^\rstar\dc (\bar b\dc c\dc p\dc (b\dc p)^\rstar)^\rstar\dc \bar b)\dc \bar c\\
&&=&\; (b\dc p)^\rstar\dc \bar b\dc (\id\plus c\dc p\dc (b\dc p)^\rstar\dc \bar b\dc (c\dc p\dc (b\dc p)^\rstar\dc \bar b)^\rstar)\dc \bar c&\by{\axname{$\rstar$-Uni}}\\
&&=&\; (b\dc p)^\rstar\dc \bar b\dc (c\dc p\dc (b\dc p)^\rstar\dc \bar b)^\rstar\dc \bar c&\by{\axname{$\rstar$-Fix}}\\
&&=&\; (\wt{b}{p})\dc \wt{c}{(p\dc \wt{b}{p})}.
\end{flalign*}
To show~\axname{TW-Uni}, assume that $u\dc \bar b = \bar c\dc v$ and 
$u\dc b\dc p = c\dc q\dc u$ for suitable $b,c,u,v,p$ and $q$. 
and therefore
\begin{flalign*}
&& u\dc \wt{b}{p} &\;= u\dc (b\dc p)^\rstar\dc \bar b&\\
&& &\;= (c\dc q)^\rstar\dc u\dc \bar b&\by{\axname{$\rstar$-Uni}}\\
&& &\;= (c\dc q)^\rstar\dc \bar c\dc v&\\
&&  &\;= (\wt{c}{q})\dc v.
\end{flalign*}

\medskip
\noindent
(vi) %
uniform Conway $\to$ KiCT. 
By combining~\eqref{eq:rstar_from_wh} and~\eqref{eq:while-as-it}, we express
$(\argument)^\rstar$ through $(\argument)^\istar$:
\begin{flalign*}
&& p^\rstar &\;= \inr\dc (\ifd{[\zero,\inr]}{[\inl,\inl\plus p\dc \inr]\dc \tt}{\ff})^\istar\dc [\id,\div].&
\end{flalign*}
Since, $\inr\dc \ifd{[\zero,\inr]}{[\inl,\inl\plus p\dc \inr]\dc \tt}{\ff} = (\inl\plus p\dc \inr)\dc (\id\cpp\inr)$,
using \UNI, 
\begin{align*}
p^\rstar &\;= (\inl\plus p\dc\inr)^\istar\dc (\id\cpp\inr)\dc [\id,\div] = (\inl\plus p\dc\inr)^\istar.
\end{align*}
Let us show~\axname{$\rstar$-Fix} and~\axname{$\rstar$-Sum}: 
\begin{flalign*}
&& p^\rstar &\;= (\inl\plus p\dc\inr)^\istar&\\*
&&  &\;= (\inl\plus p\dc\inr)\dc [\id,p^\rstar]&\by{\FIX}\\
&&  &\;= \id\plus p\dc p^\rstar,\\[2ex]
&& \!\!(p\plus q)^\rstar &\;= (\inl\plus p\dc\inr\plus q\dc\inr)^\istar&\\
&&  &\;= ((\inll\plus p\dc\inr\plus q\dc\inrl)\dc [\id,\inr])^\istar&\\
&&  &\;= (\inll\plus p\dc\inr\plus q\dc\inrl)^{\istar\istar}&\by{\COD}\\
&&  &\;= ((\inl\plus p\dc\inr)\dc ((\inl\plus q\dc\inr)\cpp\id))^{\istar\istar}&\\
&&  &\;= ((\inl\plus p\dc\inr)^\istar\dc (\inl\plus q\dc\inr))^{\istar}&\by{\NAT}\\
&&  &\;= (p^\rstar\dc (\inl\plus q\dc\inr))^{\istar}&\\
&&  &\;= p^\rstar\dc (\inl\plus q\dc\inr)\dc [\id,(p^\rstar\dc (\inl\plus q\dc\inr))^{\istar}]&\by{\FIX}\\
&&  &\;= p^\rstar\dc (\inl\plus q\dc\inr)\dc [\id,(p^\rstar\dc \inr\dc [\inl, \inl\plus q\dc\inr])^{\istar}]&\\
&&  &\;= p^\rstar\dc ((\inl\plus q\dc\inr)\dc [\inl, p^\rstar\dc\inr])^{\istar}&\by{\DIN}\\
&&  &\;= p^\rstar\dc (\inl\plus q\dc p^\rstar\dc\inr)^{\istar}&\\
&&  &\;= p^\rstar\dc (q\dc p^\rstar)^{\rstar}.&
\intertext{
Finally, let us show~\axname{$\rstar$-Uni}. Assume that $u\dc p = q\dc u$, which entails 
$u\dc (\inl\plus p\dc\inr) = (q\dc\inr\plus u\dc\inl)\dc (\id\cpp u)$}
&& u\dc p^\rstar &\;= u\dc (\inl\plus p\dc\inr)^\istar&\\
&&  &\;= (q\dc\inr\plus u\dc\inl)^\istar&\by{\UNI}\\
&&  &\;= ((\inl\plus q\dc\inr)\dc (u\cpp\id))^\istar\\
&&  &\;= (\inl\plus q\dc\inr)^\istar\dc u&\by{\NAT}\\
&&  &\;= q^\rstar\dc u&\text{\qed}
\end{flalign*}
\subsection{Proof of \autoref{the:resmp}}
First, observe that $\BV_{\BBT_H}$ is an idempotent grove category under the following 
semilattice structure on $\BV(X,T_HY)$: the bottom element is $\zero\dc\out^\mone$,
and the join of $f,g\c {X\to T_HY}$ is $(f\dc\out\plus g\dc\out)\dc \out^\mone$.

For every $f\in\BV(X,TY)$, $f\dc T\inl\dc\out^\mone\in\BV(X,T_HY)$, which yields 
a functor $J\c\BV_{\BBT}\to\BV_{\BBT_H}$. The action of $J$ on morphisms is monomorphic 
-- moreover, it is a section: for every $g\in\BV(X,T_HY)$, let $Rg = g\dc\out\dc [\eta,\zero]^\klstar\in\BV(X,TY)$; then
\begin{align*}
R(Jf) 
 =&\; f\dc T\inl\dc\out^\mone\dc\out\dc [\eta,\zero]^\klstar\\*
 =&\; f\dc T\inl\dc [\eta,\zero]^\klstar\\
 =&\; f\dc (\inl\dc [\eta,\zero])^\klstar\\
 =&\; f.
\end{align*}
Thus, $\BV_{\BBT}$ is a wide subcategory of $\BV_{\BBT_H}$. Since $\Tame$ is a 
wide subcategory of $\BV_{\BBT}$, by transitivity, $\Tame$ is a 
wide subcategory of $\BV_{\BBT_H}$. We proceed to prove that ${(\Tame, \BV_{\BBT_H})}$
is a KiC. By \autoref{thm:elgot-grove}, we equivalently prove that
$\BV_{\BBT_H}$ supports $\Tame$-uniform Conway iteration.

It is already known~\cite[Lemma 7.2]{GoncharovSchroderEtAl18} that if $\BBT$ supports Conway iteration 
then so does~$\BBT_H$. 
By \autoref{thm:elgot-grove}, we 
are left to check that $\BBT_H$ satisfies $\UNI$. To that end, we recall how 
Elgot iteration~$(\argument)^\iistar$ for $\BBT_H$ is defined via Elgot 
iteration~$(\argument)^\iistar$ for $\BBT$. Given $f\c X\to T_H(Y\cpp X)$, 
$f^\iistar\c X\to T_HY$ is the unique solution of the equation 
\begin{align}\label{eq:uni-fix}
f^\iistar\dc\out =&\; (f\dc\out\dc T\pi)^\istar\dc T(\id\cpp H[\hat\eta,f^\iistar]^{\widehat\klstar})
\end{align}
where $\pi$ is the natural isomorphism $[\inl\cpp\id,\inrl]\c (A\cpp B)\cpp C\to (A\cpp C)\cpp B$,
and~$\hat\eta$ and~$(\argument)^{\widehat\klstar}$ are unit and Kleisli lifting 
of $\BBT_H$. Suppose that for some $f\c X\to T_H(Y\cpp X)$, $f\c Z\to T_H(Y\cpp Z)$
and for tame $u\c X\to T_H Z$,
\begin{align}\label{eq:uni-assm}
u\dc g^{\widehat\klstar} =&\; f\dc [\inl\dc\hat\eta, u\dc T_H\inr]^{\widehat\klstar},
\intertext{
and show that 
}
u\dc (g^{\iistar})^{\widehat\klstar} =&\; f^{\iistar}.\notag
\end{align}
We do this essentially by modifying the existing argument~\cite[Lemma 7.2]{GoncharovSchroderEtAl18}. 
That~$u$ is tame entails that $u\dc\out=u'\dc T\inl$ for some $u'$.
We obtain
\begin{flalign*}
&& u'\dc (g\dc&\out\dc T\pi)^\klstar\\*
&&  =&\; u'\dc T\inl\dc [g\dc\out,Hg^{\widehat\klstar}\dc\inr\dc\eta]^\klstar\dc T\pi\\* 
&&  =&\; u\dc\out\dc [g\dc\out,Hg^{\widehat\klstar}\dc\inr\dc\eta]^\klstar\dc T\pi\\
&&  =&\; u\dc g^{\widehat\klstar}\dc\out\dc T\pi\\
&&  =&\; f\dc [\inl\dc\hat\eta, u\dc T_H\inr]^{\widehat\klstar}\dc\out\dc T\pi&\by{\eqref{eq:uni-assm}}\\
&&  =&\; f\dc\out\dc [[\inl\dc\hat\eta, u\dc T_H\inr]\dc\out,\\&&&\qquad\qquad H[\inl\dc\hat\eta, u\dc T_H\inr]^{\widehat\klstar}\dc\inr\dc\eta]^\klstar\dc T\pi\\
&&  =&\; f\dc\out\dc [[\inll\dc\eta, u'\dc T\inrl],\\&&&\qquad\qquad H[\inl\dc\hat\eta, u\dc T_H\inr]^{\widehat\klstar}\dc\inr\dc\eta]^\klstar\dc T\pi\\
&&  =&\; f\dc\out\dc [[\inll\dc\eta, u'\dc T\inr], H[\inl\dc\hat\eta, u\dc T_H\inr]^{\widehat\klstar}\dc\inrl\dc\eta]^\klstar\\
&&  =&\; f\dc\out\dc T\pi\dc T[[\inll, H[\inl\dc\hat\eta, u\dc T_H\inr]^{\widehat\klstar}\dc\inrl], \inr]\dc\\&&&\qquad\qquad [\inl\dc\eta, u'\dc T\inr]^\klstar\\
&&  =&\; f\dc\out\dc T\pi\dc T((\id\cpp H[\inl\dc\hat\eta, u\dc T_H\inr]^{\widehat\klstar})\cpp\id)\dc\\&&&\qquad\qquad[\inl\dc\eta, u'\dc T\inr]^\klstar.
\end{flalign*}
By \UNI, this entails
\begin{align*}
u'\dc ((g\dc\out\dc T\pi)^\istar)^\klstar = (f\dc\out\dc T\pi\dc T((\id\cpp H[\inl\dc\hat\eta, u\dc T_H\inr]^{\widehat\klstar})\cpp\id))^\istar.
\end{align*}
By \NAT, furthermore:
\begin{align}\label{eq:uni-aux}
u'\dc ((g\dc\out\dc T\pi)^\istar)^\klstar = (f\dc\out\dc T\pi)^\istar\dc T(\id\cpp H[\inl\dc\hat\eta, u\dc T_H\inr]^{\widehat\klstar}).
\end{align}
Now
\begin{flalign*}
&&u\dc (g^{\iistar})^{\widehat\klstar}\dc\out 
 =&\; u'\dc T\inl\dc [g^{\iistar}\dc\out,H(g^{\iistar})^{\widehat\klstar} \dc\inr\dc\eta]^\klstar\\
&& =&\; u'\dc (g^{\iistar}\dc\out)^{\klstar}\\
&& =&\; u'\dc ((g\dc\out\dc T\pi)^\istar)^\klstar\dc T(\id\cpp H[\hat\eta,g^\iistar]^{\widehat\klstar})\\
&& =&\; (f\dc\out\dc T\pi)^\istar\dc T(\id\cpp H[\inl\dc\hat\eta, u\dc T_H\inr]^{\widehat\klstar}) \dc T(\id\cpp H[\hat\eta,g^\iistar]^{\widehat\klstar})&\by{\eqref{eq:uni-aux}}\\
&& =&\; (f\dc\out\dc T\pi)^\istar\dc T(\id\cpp H[\hat\eta,u\dc (g^{\iistar})^{\widehat\klstar}]).
\end{flalign*}
This entails $u\dc (g^{\iistar})^{\widehat\klstar} = f^{\iistar}$, for, by definition, 
$f^{\iistar}$ is uniquely determined by equation~\eqref{eq:uni-fix}.
\qed

\subsection{Proof of \autoref{lem:fix-rat}}
The dependency condition is obvious in all three cases. We will prove finiteness
of sets of derivatives only.

\textit{(Sum)} Let $t,s\in T_\nu n$ be rational. Then $\Der(t\plus s) = \{t\plus s\}\cup\Der(t)\cup\Der(s)$,
which is finite, since $\Der(t)$ and $\Der(s)$ are so. 

\textit{(Composition)} It suffices to stick to the following instance: given  
$n,k\in\nat$, a rational element $t\in T_\nu n$ and a rational map $s\c n\to T_\nu k$, show that $t\dc s^\klstar\in T_\nu k$
is rational. 
Consider the set $P$ of sums of the form 
\begin{align*}
t'\dc s^\klstar\plus \sum_{s'\in\Der(s)} r_{s'}\dc (s')^\klstar
\end{align*}
where $t'$ ranges over $\Der(t)$ and $r_{s'}$ range over 
non-guarded elements of $T_\nu n$. Then $P$ is finite.
Moreover, $P$ is closed under derivatives: using~\autoref{lem:D-comp}, 
\begin{flalign*}
&& \partial_{b,\us}(t'&\dc s^\klstar\plus \sum_{s'\in\Der(s)} r_{s'}\dc (s')^\klstar) \\
&&&\;= \partial_{b,\us}(t')\dc s^\klstar\plus o(t')\dc (\partial_{b,\us}(s))^\klstar\plus \sum_{s'\in\Der(s)}\partial_{b,\us}(r_{s'})\dc (s')^\klstar\\*
&&&\qquad \plus \sum_{s'\in\Der(s)} o(r_{s'})\dc (\partial_{b,\us}(s'))^\klstar  &\\
&&&\;= \partial_{b,\us}(t')\dc s^\klstar\plus o(t')\dc (\partial_{b,\us}(s))^\klstar\plus \sum_{s'\in\Der(s)} o(r_{s'})\dc (\partial_{b,\us}(s'))^\klstar  &\\
&&&\;= \partial_{b,\us}(t')\dc s^\klstar\plus (o(t')\plus o(r_s))\dc (\partial_{b,\us}(s))^\klstar\\*
&&&\qquad \plus \sum_{s'\in\Der(s)\smin\{s\}} o(r_{s'})\dc (\partial_{b,\us}(s'))^\klstar,  &
\end{flalign*}
and analogously for $\partial_{b,\fs}^k$. Note that $t\dc s^\klstar\in P$. Therefore 
$\Der(t\dc s^\klstar)\subseteq P$. Since $P$ is finite, so is~$\Der(t\dc s^\klstar)$.

\textit{(Iteration)} Let $t=[t_0,\ldots,t_{n-1}]\c n\to T_\nu n$, and $\Der(t_i)$ 
be finite for $i=0,\ldots,n-1$. Analogously to the previous clause, consider the set 
$P$ of sums of the form 
\begin{align*}
\sum_{t'\in\Der(t)} r_{t'}\dc (t')^\klstar\dc (t^\rstar)^\klstar + r
\end{align*}
where $t'$ ranges over $\Der(t)$ and $r_{s'}$, $r$ range over those
non-guarded elements of~$T_\nu n$. In the same manner as in
the previous clause: $P$ is finite, contains $t^\rstar$ and is closed under derivatives, 
hence $\Der(t^\rstar)$ is finite.\qed

\subsection{Proof of \autoref{pro:kleene}}

We assume that $n>0$.

($\Rightarrow$) Suppose that $[t_0,\ldots,t_{n-1}]\c n\to T_\nu k$ is rational, hence every~$\Der(t_i)$ is finite.
Let $m=\max\bigl\{n, |\bigcup_i\Der(t_i)|\bigr\}$ and if $m>n$, let us name $t_n,\ldots,t_{m-1}$
the elements in $\bigcup_i\Der(t_i)\smin\{t_0,\ldots,t_{n-1}\}$, thus $\bigcup_i\Der(t_i)=\{t_0,\ldots,t_{m-1}\}$.
Let $\hash\c\bigcup_i\Der(t_i)\to m$ send every $t_i$ to any index $j$, for which 
$t_i=t_j$.
Then we define $s_0,\ldots,s_{m-1}\in T_\nu m$ as follows:
\begin{align*}
s_{i} = \sum_{b\in\hat\Theta,\us\in\Gamma} b.\,\us.\,\hash\partial_{b,\us}(t_i)+\sum_{b\in\hat\Theta,\fs\in\Sigma} b.\,\fs(\hash\partial_{b,\fs}^1(t_i),\ldots,\hash\partial_{b,\fs}^{n_\fs}(t_i)) %
\end{align*}
where $n_\fs$ is the arity of $\fs$. Note that the involved indexed 
sums contain finitely many distinct elements, since~$t_i$ are assumed to depend only on a finite subset of $\Sigma\cup\Gamma$.
Hence, every $s_i$ is prefinite. Moreover, by definition, every $s_i$ is flat and 
guarded. It suffices therefore to prove that $t_i = \inj_i\dc s^\rstar\dc (o(t))^\klstar$
for $i=0,\ldots,n-1$ where $s=[s_0,\ldots,s_{m-1}]$ and $t=[t_0,\ldots,t_{m-1}]$.
Let 
\begin{align*}
	\Bis=\{(\inj_{i}\dc s^\rstar\dc (o(t))^\klstar,~t_i)\mid i=0,\ldots,m-1\}.
\end{align*}
This relation satisfies the conditions of \autoref{lem:bisim}. Let us verify~(1)
and (3) -- the remaining condition (2) are handled analogously. We have
\begin{flalign*}
&&\partial_{b,\us}(\inj_{i}\dc s^\rstar&\dc (o(t))^\klstar) \\*
&&=&\;\partial_{b,\us}(\inj_{i}\dc (\eta\plus s\dc (s^\rstar)^\klstar)\dc (o(t))^\klstar)&\by{\axname{$\rstar$-Fix}}\\
&&=&\;\partial_{b,\us}(\inj_{i}\dc o(t)\plus \inj_{i}\dc s\dc (s^\rstar)^\klstar\dc (o(t))^\klstar)\\
&&=&\;\partial_{b,\us}(s_i\dc (s^\rstar\dc (o(t))^\klstar)^\klstar)\\
&&=&\;\partial_{b,\us}(s_i)\dc (s^\rstar\dc (o(t))^\klstar)^\klstar\plus o(s_i)\dc (\partial_{b,\us}(s^\rstar\dc (o(t))^\klstar))^\klstar&\by{\autoref{lem:D-comp}}\\
&&=&\;\partial_{b,\us}(s_i)\dc (s^\rstar)^\klstar\dc (o(t))^\klstar\\
&&=&\; \inj_{\hash\partial_{b,\us}(t_i)}\dc s^\rstar\dc (o(t))^\klstar
\intertext{and $(\inj_{\hash\partial_{b,\us}(t_i)}\dc s^\rstar\dc (o(t))^\klstar,~\partial_{b,\us}(t_i))\in\Bis$, since $\partial_{b,\us}(t_i)=t_{\hash\partial_{b,\us}(t_i)}$.}
&&o(\inj_{i}\dc s^\rstar\dc (&o(t))^\klstar) \\*
&&=&\;o(\inj_{i}\dc (\eta\plus s\dc (s^\rstar)^\klstar)\dc (o(t))^\klstar)&\by{\axname{$\rstar$-Fix}}\\
&&=&\;o(\inj_{i}\dc o(t)\plus \inj_{i}\dc s\dc (s^\rstar)^\klstar\dc (o(t))^\klstar)\\
&&=&\;o(o(t_i)\plus s_i\dc (s^\rstar)^\klstar\dc (o(t))^\klstar)\\
&&=&\;o(t_i).
\end{flalign*}
Now, by \autoref{lem:bisim}, $t = \inl\dc s^\rstar\dc (o(t))^\klstar$, as desired.

($\Leftarrow$) Suppose that $t\c n\to T_\nu k$ is definable, hence it has the form
$\inj_{n,m}\dc s^\rstar\dc r^\klstar$ for some rational~$s$ and $r$. The claim then 
follows from \autoref{lem:fix-rat}.\qed

\subsection{Proof of \autoref{the:free}}
We rely on the following observation.
\begin{lemma}\label{lem:funct}
Let $(\BC,\Tame)$ and $(\BD,\Tame[\BD])$ be two KiCs, and let $F$ be the following 
map, acting on objects and on morphisms: $FX\in |\BD|$ for every $X\in |\BC|$, 
$Fp\in\BD(FX,FY)$ for every $p\in\BC(X,Y)$. Suppose that $F$ preserves coproducts, $Fp\in\Tame[\BD](FX,FY)$ for all $p\in\Tame(X,Y)$,
and the following further preservation properties hold
\begin{align*}
Fp^\rstar = (Fp)^\rstar,\quad F(p\dc\inl) = Fp\dc\inl,\quad F(p\dc\inr) = Fp\dc\inr,\quad F(p\dc [\zero,\id]) = Fp\dc [\zero,\id].
\end{align*}
Then $F$ is a KiC-functor.
\end{lemma}
\begin{proof}
Note that the assumptions entail that $\id$ and $\zero$ are also preserved 
by $F$: $F\id = F(\inl\dc [\id,\zero]) = \inl\dc [\id,\zero] = \id$, $F\zero = F(\inr\dc [\id,\zero]) = \inr\dc [\id,\zero] = \zero$.
We need to show that the assumptions entail that $F$ preserves nondeterministic sums and morphism 
compositions. 

\textit{(Compositions)} Let $p\in\BC(X,Y)$, $q\in\BC(Y,Z)$, and show that $F(p\dc q) = Fp\dc Fq$.
Note that, essentially by \axname{$\rstar$-Fix},  
\begin{align*}
\inll\dc [[p\dc&\,\inrl, q\dc\inr],\zero]^\rstar\dc [\zero,\id] \\
=\;& \inll\dc [\zero,\id]\plus \inll\dc [[p\dc\inrl, q\dc\inr],\zero]\dc [[p\dc\inrl, q\dc\inr],\zero]^\rstar\dc [\zero,\id]\\
=\;& p\dc\inrl\dc [[p\dc\inrl, q\dc\inr],\zero]^\rstar\dc [\zero,\id]\\
=\;& p\dc(\inrl\dc [\zero,\id]\plus \inrl\dc[[p\dc\inrl, q\dc\inr],\zero]\dc [[p\dc\inrl, q\dc\inr],\zero]^\rstar\dc [\zero,\id])\\
=\;& p\dc q\dc\inr\dc [[p\dc\inrl, q\dc\inr],\zero]^\rstar\dc [\zero,\id]\\
=\;& p\dc q\dc(\inr\dc [\zero,\id]\plus \inr\dc[[p\dc\inrl, q\dc\inr],\zero]\dc [[p\dc\inrl, q\dc\inr],\zero]^\rstar\dc [\zero,\id])\\
=\;& p\dc q,
\intertext{and hence}
F(p\dc q)
=\;& F(\inll\dc [[p\dc\inrl, q\dc\inr],\zero]^\rstar\dc [\zero,\id])\\
=\;& \inll\dc [[Fp\dc\inrl, Fq\dc\inr],\zero]^\rstar\dc [\zero,\id])\\
=\;& Fp\dc Fq.
\end{align*}

\textit{(Sums)} Note that, essentially by \axname{$\rstar$-Fix}, $\inl\dc [\inr,\zero]^\rstar = \inl\dc(\id\plus [\inr,\zero]\dc [\inr,\zero]^\rstar) = \inl\plus \inr\dc [\inr,\zero]^\rstar 
= \inl\plus \inr\dc (\id\plus [\inr,\zero]\dc [\inr,\zero]^\rstar) = \inl\plus \inr$. Thus, 
for any $p,q\in\BC(X,Y)$, using the previous clause, 
\begin{align*}
F(p\plus q) 
=\;& F((\inl\plus\inr)\dc [p, q])\\  
=\;& F(\inl\dc [\inr,\zero]^\rstar\dc [p, q])\\  
=\;& \inl\dc [\inr,\zero]^\rstar\dc [Fp, Fq]\\  
=\;& (\inl\plus\inr)\dc [Fp, Fq]\\
=\;& Fp\plus Fq,
\end{align*}
and we are done.
\end{proof}
Let us return to the proof of \autoref{the:free}.
The key observation is that $\sem{\inj_{n,m}\dc s^\rstar\dc r^\klstar}_\shortuparrow$
does not depend on the choice of $s$ and $r$. This is argued as follows. Using 
the construction in \autoref{pro:kleene}, for a given $t=\inj_{n,m}\dc s^\rstar\dc r^\klstar$, 
we obtain a canonical representation $t=\inj_{n,l}\dc \hat s^\rstar\dc\hat r^\klstar$,
with $\hat s\c l\to T_\nu l$, $\hat r\c l\to T_\nu k$, and this representation only depends on $t$, hence, it suffices to show that 
\begin{align}\label{eq:uni-rep}
	\sem{\inj_{n,m}\dc s^\rstar\dc r^\klstar}_\shortuparrow = \sem{\inj_{n,l}\dc \hat s^\rstar\dc\hat r^\klstar}_\shortuparrow.
\end{align}
Because of the restrictions on $s$ and $r$, there is an epimorphism $u\c m\to l$,
such that $s\dc T_\nu u = u\dc\hat s$ and $u\dc \hat r = r$. W.l.o.g.\ assume that 
$\inj_{l,m}$ is a left inverse of $u$. 
Now, \eqref{eq:uni-rep} is obtained as follows:
\begin{flalign*}
&& 	 \sem{\inj_{n,m}\dc s^\rstar\dc r^\klstar}_\shortuparrow 
=\;& \inj_{n,m}\dc \sem{s}^\rstar\dc \sem{r}\\
&&=\;& \inj_{n,m}\dc \sem{s}^\rstar\dc\sem{u} \dc\sem{\hat r}\\
&&=\;& \inj_{n,m}\dc u\dc\sem{\hat s}^\rstar\dc\sem{\hat r}&\by{\axname{$\rstar$-Uni}}\\
&&=\;& \inj_{n,l}\dc\inj_{l,m}\dc u\dc\sem{\hat s}^\rstar\dc\sem{\hat r}\\
&&=\;& \inj_{n,l}\dc \sem{\hat s}^\rstar\dc\sem{\hat r}\\
&&=\;& \sem{\inj_{n,l}\dc \hat s^\rstar\dc\hat r^\klstar}_\shortuparrow
\end{flalign*}
The defined lifting $\sem{\argument}_\shortuparrow\c \Frep\to\BC$ is easily seen to 
make~\eqref{eq:free-cpp} commute. Also, note that there is no more than one structure-preserving 
candidate for $\sem{\argument}_\shortuparrow$, to make~\eqref{eq:free-cpp} commute: indeed, 
since every morphism in $\Frep$ is representable as $t=\inj_{n,m}\dc s^\rstar\dc r^\klstar$,
$\sem{t}_\shortuparrow$ must only be defined as $\inj_{n,m}\dc \sem{s}^\rstar\dc \sem{r}$.

We are left to check that $\sem{\argument}_\shortuparrow$ is a KiCT-functor, which is 
facilitated by \autoref{lem:funct}. The only non-trivial clause is 
preservation of Kleene star. As an auxiliary step, we show that 
\begin{align}\label{eq:1plus}
\sem{\eta\plus\inj_{n,m}\dc s^\rstar\dc r^\klstar}_\shortuparrow = \sem{\eta}_\shortuparrow\plus \sem{\inj_{n,m}\dc s^\rstar\dc r^\klstar}_\shortuparrow
\end{align}
for any guarded flat $s\c m\to T_\nu m$, and a non-guarded $r\c m\to T_\nu n$.
In order to calculate the left-hand side of~\eqref{eq:1plus}, we need to 
find a suitable representation for $\id\plus\inj_{n,m}\dc s^\rstar\dc r^\klstar$.
Concretely, we show that
\begin{align*}
\eta\plus\inj_{n,m}\dc s^\rstar\dc r^\klstar = \inj_{n,m+m}\dc [\inr\dc\eta,s\dc T_\nu\inr]^\rstar\dc [[\eta_n,\zero],r]^\klstar
\end{align*}
Indeed, using \axname{$\rstar$-Fix} and \axname{$\rstar$-Uni},
\begin{align*}
\inj_{n,m+m}\dc &[\inr\dc\eta,s\dc T_\nu\inr]^\rstar\dc [[\eta_n,\zero],r]^\klstar\\
=\;& \inj_{n,m+m}\dc (\eta\plus [\inr\dc\eta,s\dc T_\nu\inr]\dc ([\inr\dc\eta,s\dc T_\nu\inr]^\rstar)^\klstar)\dc [[\eta_n,\zero],r]^\klstar\\
=\;& \eta\plus \inj_{n,m}\dc \inr\dc\eta\dc ([\inr\dc\eta,s\dc T_\nu\inr]^\rstar)^\klstar\dc [[\eta_n,\zero],r]^\klstar\\
=\;& \eta\plus \inj_{n,m}\dc s^\rstar\dc T_\nu\inr\dc [[\eta_n,\zero],r]^\klstar\\
=\;& \eta\plus \inj_{n,m}\dc s^\rstar\dc r^\klstar.
\end{align*}
Now, \eqref{eq:1plus} turns into 
\begin{align*}
\inj_{n,m+m}\dc [\inr,\sem{s}\dc\inr]^\rstar\dc [[\id_n,\zero],\sem{r}] = \id\plus \inj_{n,m}\dc \sem{s}^\rstar\dc \sem{r}.
\end{align*}
This equation is shown as above, since \axname{$\rstar$-Fix} and \axname{$\rstar$-Uni}
are sound for $\BC$.

An analogous method is used to show that $\sem{\argument}_\shortuparrow$ preserves Kleene star.
Let $s\c m\to T_\nu m$ be guarded flat, and let $r\c m\to T_\nu n$ be non-guarded,
and prove that:
\begin{align}\label{eq:star_sem}
\sem{(\inj_{n,m}\dc s^\rstar\dc r^\klstar)^\rstar}_\shortuparrow = \sem{\inj_{n,m}\dc s^\rstar\dc r^\klstar}_\shortuparrow^\rstar.
\end{align}
The following equation is provable using the axioms of KiC 
\begin{align*}
(\inj_{n,m}\dc s^\rstar\dc r^\klstar)^\rstar = \eta\plus\inj_{n,m}\dc ((r\dc T_\nu\inj_{n,m})^\rstar\dc s^\klstar)^\rstar\dc ((r\dc T_\nu\inj_{n,m})^\rstar\dc r^\klstar)^\klstar,
\end{align*}
hence, the equation
\begin{align*}
(\inj_{n,m}\dc \sem{s}^\rstar\dc\sem{r}^\klstar)^\rstar = \id\plus\inj_{n,m}\dc ((\sem{r}\dc\inj_{n,m})^\rstar\dc \sem{s})^\rstar\dc (\sem{r}\dc\inj_{n,m})^\rstar\dc \sem{r}
\end{align*}
is provable as well. Now, the proof of \eqref{eq:star_sem} is as follows:
\begin{flalign*}
&&\!\!\!\!\sem{(\inj_{n,m}\dc s^\rstar\dc r^\klstar)^\rstar}_\shortuparrow
=\;& \sem{\eta\plus\inj_{n,m}\dc ((r\dc T_\nu\inj_{n,m})^\rstar\dc s^\klstar)^\rstar\dc ((r\dc T_\nu\inj_{n,m})^\rstar\dc r^\klstar)^\klstar}_\shortuparrow\\
&&=\;&\id\plus\inj_{n,m}\dc ((\sem{r}\dc\inj_{n,m})^\rstar\dc \sem{s}^\klstar)^\rstar\dc (\sem{r}\dc\inj_{n,m})^\rstar\dc \sem{r}&\by{\eqref{eq:1plus}}\\
&&=\;& (\inj_{n,m}\dc \sem{s}^\rstar\dc\sem{r}^\klstar)^\rstar\\
&&=\;& \sem{\inj_{n,m}\dc s^\rstar\dc r^\klstar}_\shortuparrow^\rstar.&\qed
\end{flalign*}

\end{document}